\documentclass{llncs}

\usepackage{hyperref}

\usepackage[normalem]{ulem}

\usepackage{times,multirow}
\usepackage{stmaryrd}
\usepackage{amssymb}
\usepackage{verbatim}
\usepackage{amsmath,amsfonts}
\usepackage{graphicx}
\usepackage{color}
\usepackage{xcolor,pgf,tikz,pgflibraryarrows,pgffor,pgflibrarysnakes}
\usepackage{microtype}
\usepackage{gastex,graphicx}
\usepackage{amsmath}
\usepackage{amsfonts}
\usepackage{amssymb}
\usepackage{tikz}

\usetikzlibrary{fit} 
\usetikzlibrary{backgrounds} 

\usetikzlibrary{calc}
\tikzstyle{RectObject}=[rectangle,fill=white,draw,line width=0.5mm]
\tikzstyle{line}=[draw]
\tikzstyle{arrow}=[draw, -latex] 
\newcommand{\cnt}{\mathsf{C}}

\newcommand{\spl}{\mathsf{Spoiler}}
\newcommand{\dpl}{\mathsf{Duplicator}}

\newcommand{\N}{\mathbb N}

\newcommand{\wB}{\Box^{\mathsf{w}}}
\newcommand{\wU}{\until^{\mathsf{w}}}
\newcommand{\wF}{\fut^{\mathsf{w}}}

\newcommand{\until}{\:\mathsf{U}}
\newcommand{\since}{\:\mathsf{S}}

\newcommand{\R}{\:\mathbb{R}}
\newcommand{\fut}{\Diamond}

\mathchardef\mhyphen="2D
\mathchardef\mhyph="2D

\newcommand{\nex}{\mathsf{O}}
\newcommand{\nx}{\mathsf{O}}

\newcommand{\mtl}{\mathsf{MTL}}

\newcommand{\oomit}[1]{}

\title{Metric Temporal Logic with Counting}
\author{ Khushraj Madnani\inst{1}, Shankara Narayanan Krishna\inst{1}  \and Paritosh K. Pandya \inst{2}}
\institute{Department of Computer Science \& Engineering, 
IIT Bombay, Mumbai, India 400 076. \\
\email{\{khushraj,krishnas\}@cse.iitb.ac.in}\\
\and
School of Technology and Computer Science,
Tata Institute of Fundamental Research, \\
Mumbai, India 400 005. \\
\email{pandya@tcs.tifr.res.in}}

\begin{document}
\maketitle

Ability to count number of occurrences of events within a specified time
interval is very useful in specification of resource bounded real time
computation. In this paper, we study an extension of Metric Temporal Logic
($\mathsf{MTL}$) with two different counting modalities called $\mathsf{C}$ and $\mathsf{UT}$ (until
with threshold), which enhance the expressive power of $\mathsf{MTL}$ in orthogonal
fashion. We confine ourselves only to the future fragment of $\mathsf{MTL}$
interpreted in a pointwise manner over finite timed words. We provide a
comprehensive study of the expressive power of logic $\mathsf{CTMTL}$ and its
fragments using the technique of EF games extended with suitable counting
moves. Finally, as our main result, we establish the decidability of
$\mathsf{CTMTL}$ by giving an equisatisfiable reduction from $\mathsf{CTMTL}$ to $\mathsf{MTL}$. The
reduction provides one more example of the use of temporal projections with
oversampling introduced earlier for proving decidability.
Our reduction also implies that  $\mathsf{MITL}$ extended with $\mathsf{C}$ and $\mathsf{UT}$ modalities
is elementarily decidable.

\section{Introduction}
 Temporal logics provide constructs to specify qualitative ordering between events in time. But real time logics have the ability to specify quantitative timing
constraints between events. Metric Temporal Logic $\mtl$ is amongst the best studied of real time logics. Its principle modality $a \until_I b$ states that
an event  $b$ should occur in future within a time distance lying within interval $I$. Moreover, $a$ should hold continuously till then. 

    In many situations, especially those dealing with resource bounded computation, the ability to count the number of occurrences of events 
becomes important. In this paper, we consider an extension of $\mtl$ with two counting modalities $\cnt$ and $\mathsf{UT}$ (until threshold) which provide differing 
abilities to specify constraints on counts on events in time intervals. The resulting logic is called $\mathsf{CTMTL}$. Modality $\cnt_I^{\geq n} ~ \phi$ states that 
the number of times formula $\phi$ holds in time interval $I$ (measured relative to current time point) is at least $n$. 
This is a mild generalization of $\cnt_{(0,1)}^{\geq n} ~ \phi$ modality studied by Rabinovich \cite{rabin} in context of continuous time $\mtl$. 
The $\mathsf{UT}$  modality $\phi ~\until_{I,\#\kappa \geq n} ~\psi$ is like $\mtl$ until but it 
additionally states that the number of time formula $\kappa$ holds between now and time point where $\psi$ holds is at least $n$. Thus it extends $\until$ to 
simultaneously specify constraint on time and count of subformula. Constraining $\until$ by count of subformula was already explored for untimed $\mathsf{LTL}$ 
by Laroussini et al \cite{cltl}. But the combination of timing and counting seems new. The following example illustrates the use of these modalities.\\
\noindent \emph{An Example}. We specify some constraints to be monitored by exercise bicycle electronics.
\begin{itemize}
 \item 
Two minutes after the start of exercise, the heartbeat (number of pulses in next 60 seconds) should be between 90 and 120.
This can be stated as \\ 
                  $\Box(st \Rightarrow (\cnt_{[120,180]}^{\geq 90} pulse \land \cnt_{[120,180]}^{< 120} pulse))$    
\item Here is one exerise routine:
After start of exercise, $slow\_peddling$ should be done for 1 kilometre (marked by odometer giving 1000 pulses) and this should be achieved in interval 1 to 2 minutes.
After this $fast\_peddling$ should be done for 3 minutes. This can be specified as
$ \Box(st \Rightarrow slowpeddle ~\until_{[60,120], \#odo = 1000}  ~(\Box_{[0,180]} fastpeddle))$
\end{itemize}
The expressiveness and decidability properties of real time logics differ considerable based on nature of time. There has been considerable study of
counting $\mtl$ in continuous time \cite{rabin-tcs},\cite{hunter}.  In this paper, we consider the case of pointwise time, i.e. $\mathsf{CTMTL}$ interpreted over finite timed words in a pointwise manner. We provide a comprehensive picture of expressiveness and decidability of $\mathsf{CTMTL}$ and its fragments in pointwise time and we find 
that this differs considerably when compared with continuous time. 

As our first main result, we show that the $\cnt$ and the $\mathsf{UT}$ modalities both increase the expressive power of $\mtl$ but they are mutually incomparable.
EF games are a classical technique used to study expressive power of logic. \cite{concur11} have adapted EF games to $\mtl$ and shown a number of expressiveness
results. In this paper, we extend $\mtl$ EF games with counting moves corresponding to the $\cnt$ and $\mathsf{UT}$ modalities. We use the resulting EF theorem
to characterise expressive  powers of several fragments of $\mathsf{CTMTL}$.

One attraction of pointwise $\mtl$  over finite timed words is that its satisfiability is decidable \cite{Ouaknine05} whereas 
continuous time $\mtl$ has  undecidable satisfiability. As our second main result, we show that $\mtl$ extended with $\cnt$ and $\mathsf{UT}$ modalities also has
decidable satisfiability. In order to prove this result, we give an equisatisfiable reduction from $\mathsf{CTMTL}$ to $\mtl$. The reduction makes use of the
notion of temporal projections modulo oversampling introduced earlier \cite{time14} where timed words satisfying original $\mathsf{CTMTL}$ formula have to be oversampled
with additional time points to satisfy corresponding $\mtl$ formula. This result marks one more use of the technique of temporal projections.
We note that our reduction can also be applied to $\mathsf{MITL}$ (with both $\until$ and $\since$) extended with $\cnt$ and $\mathsf{UT}$ and it it gives an equisatisfiable formula
in $\mathsf{MITL}$ which is exponential in the size of original formula. Thus, we establish that $\mathsf{CTMITL}[\until,\since]$ has elementary satisfiability.

\section{A Zoo of Timed Temporal Logics}
\label{zoo}
\vspace{-0.1cm}
In this section, we present the syntax and semantics of the various timed temporal logics we study in this paper. 
  Let $\Sigma$ be a finite set of propositions. A finite timed word over $\Sigma$ is a tuple
$\rho = (\sigma,\tau)$.   $\sigma$ and $\tau$ are sequences $\sigma_1\sigma_2\ldots\sigma_n$ and  $t_1t_2\ldots t_n$ respectively, with $\sigma_i \in 2^{\Sigma}-\emptyset$,  and $t_i \in \R_{\geq 0}$
 for $1 \leq i \leq n$ and $\forall i \in dom(\rho)$,  $t_i \le t_{i+1}$, where $dom(\rho)$ is the set of positions $\{1,2,\ldots,n\}$ in the timed word. 
An example of a timed word 
 over $\Sigma=\{a,b\}$ is $\rho=(\{a,b\},0.3)(\{b\}, 0.7)(\{a\},1.1)$.
$\rho$ is strictly monotonic iff $t_i < t_{i+1}$ for all $i,i+1 \in dom(\rho)$. 
Otherwise, it is weakly monotonic.
The set of finite timed words over $\Sigma$ is denoted $T\Sigma^*$.

The logic $\mathsf{MTL}$ extends  linear temporal logic ($\mathsf{LTL}$) by adding timing constraints 
to the ``until'' modality of $\mathsf{LTL}$. We parametrize this logic by a permitted 
set of time intervals denoted by $I\nu$. 
 The intervals in $I\nu$ can be  open, half-open or closed,  
with end points in $\N \cup \{0,\infty\}$.  Such an interval is denoted $\langle a, b \rangle$. For example, 
$[3,7), [5, \infty)$. Let $t+\langle a, b\rangle=\langle t+a, t+b\rangle$.
\subsection*{Metric Temporal Logic}
\label{prelim}
 Given $\Sigma$,  the formulae of $\mathsf{MTL}$ are built from $\Sigma$  using boolean connectives and 
time constrained version of the modality $\until$ as follows:
$\varphi::=a (\in \Sigma)~|true~|\varphi \wedge \varphi~|~\neg \varphi~|
~\varphi \until_I \varphi$\\
where  $I \in I\nu$.    
\label{point}
For a timed word $\rho=(\sigma, \tau) \in T\Sigma^*$, a position 
$i \in dom(\rho)$, and an $\mathsf{MTL}$ formula $\varphi$, the satisfaction of $\varphi$ at a position $i$ 
of $\rho$ is denoted $(\rho, i) \models \varphi$, and is defined as follows:\\
\noindent $\rho, i \models a$  $\leftrightarrow$  $a \in \sigma_{i}$ and~~~~~
$\rho,i  \models \neg \varphi$ $\leftrightarrow$  $\rho,i \nvDash  \varphi$\\
$\rho,i \models \varphi_{1} \wedge \varphi_{2}$   $\leftrightarrow$ 
$\rho,i \models \varphi_{1}$ 
and $\rho,i\ \models\ \varphi_{2}$\\
$\rho,i\ \models\ \varphi_{1} \until_{I} \varphi_{2}$  $\leftrightarrow$  $\exists j > i$, 
$\rho,j\ \models\ \varphi_{2}, t_{j} - t_{i} \in I$, and  $\rho,k\ \models\ \varphi_{1}$ $\forall$ $i< k <j$

\noindent $\rho$ satisfies $\varphi$ denoted $\rho \models \varphi$ 
iff $\rho,1 \models \varphi$. Let $L(\varphi)=\{\rho \mid \rho, 1 \models \varphi\}$
denote the language of a $\mathsf{MTL}$ formula $\varphi$.  
Two formulae $\varphi$ and $\phi$ are said to be equivalent denoted as $\varphi \equiv \phi$ iff $L(\varphi) = L(\phi)$.
Additional temporal connectives are defined in the standard way: 
we have the constrained future eventuality operator $\fut_I a \equiv true \until_I a$ 
and its dual 
$\Box_I a \equiv \neg \fut_I \neg a$.
We also define the next operator as $\nex_I \phi \equiv \bot \until_I  \phi$. 
Weak versions of  operators 
are defined as  $\wF_I a=a \vee \fut_I a, 
\wB_I a\equiv a \wedge \Box_I a$, $a \wU_I b\equiv b \vee [a \wedge (a \until_I b)]$ if $0 \in I$, and 
$[a \wedge (a \until_I b)]$ if $0 \notin I$.
 \vspace{-.2cm}
 
\begin{theorem}
	Satisfiability checking of $\mathsf{MTL}$ is decidable over  finite timed words and is non-primitive recursive. \cite{Ouaknine05}.
\end{theorem}
 
 \subsection*{Metric Temporal Logic with Counting ($\mathsf{CTMTL}$)}
 \label{sec:ctmtl}
 We denote by $\mathsf{CTMTL}$ the logic obtained by extending $\mathsf{MTL}$  
with the ability to count, by endowing two counting modalities $\mathsf{C}$ as well as $\mathsf{UT}$. \\
\noindent \emph{Syntax of $\mathsf{CTMTL}$}: 
 $\varphi::=a (\in \Sigma)~|true~|\varphi \wedge \varphi~|~\neg \varphi~|~\varphi~|~\cnt^{\geq n}_I\varphi~|~\varphi \until_{I,\eta} \varphi$, 
 where   $I \in I\nu$, $n \in \mathbb{N} \cup \{0\}$ and $\eta$ is a {\it threshold  formula} of the form
 $\#\varphi \geq n$ or $\#\varphi < n$.   
The counting modality
$\cnt^{\geq n}_I\varphi$ is called the $\mathsf{C}$ modality, while 
$\varphi \until_{I,\eta} \varphi$ is called the $\mathsf{UT}$ modality.
Let $\rho=(\sigma, \tau) \in T\Sigma^*$, 
$i,j \in dom(\rho)$. 
 Define \\
 \noindent $~~~~~~~~~~~N^{\rho}[i,I](\varphi)=\{k \in dom(\rho)\mid 
t_k \in t_i+I \wedge \rho, k \models \varphi\}$, and \\ 
\noindent $~~~~~~~~~~~~\rho[i,j](\varphi)=\{k \in dom(\rho) \mid i < k < j \wedge \rho, k \models \varphi\}$.\\
  Denote by  $|N^{\rho}[i,I](\varphi)|$ and $|\rho[i,j](\varphi)|$ respectively, the cardinality of $N^{\rho}[i,I](\varphi)$
  and $\rho[i,j](\varphi)$.
$|N^{\rho}[i,I](\varphi)|$ is the number of points 
in $\rho$ that lie in the interval $t_i+I$,  and which satisfy
$\varphi$, while  $|\rho[i,j](\varphi)|$ is 
the number of points lying between $i$ and $j$ which satisfy $\varphi$.
Define $\rho, i \models \cnt^{\geq n}_I \varphi$ iff   
  $|N^{\rho}[i,I](\varphi)| \geq n$.
Likewise, \noindent  $\rho, i \models \varphi_1 \until_{I,\#{\varphi} \geq n} \varphi_2$ iff  $\exists j {>} i$, $\rho, j \models \varphi_2$, ${t_j-t_i \in I}$, and  
$\rho, k \models \varphi_1$, $\forall {i < k < j}$
and $|\rho[i,j](\varphi)| \geq n$. 
\\
\noindent \emph{Remark}: 
The classical until operator of $\mathsf{MTL}$ is captured in $\mathsf{CTMTL}$ since $\varphi \until_I \psi \equiv \varphi \until_{I, \#{true} \geq 0} \psi$.
We can express $\cnt^{\sim n}_I$ and $\#{\varphi \sim n}$ for $\sim \in \{\leq, <,>, =\}$ in $\mathsf{CTMTL}$ since  
 $\cnt^{< n}_I \varphi \equiv \neg \cnt^{\ge n}_I \varphi$, $\cnt^{> n}_I \varphi \equiv \cnt^{\ge n+1}_I \varphi$, $\cnt^{\le n}_I \varphi \equiv \neg \cnt^{\ge n+1}_I \varphi$ and    
   $\#{\varphi} > n \equiv \#{\varphi} \geq n+1$, 
   $\#{\varphi} \leq n \equiv \neg (\#{\varphi} > n+1)$. 
Boolean combinations of threshold formulae are also expressible in $\mathsf{CTMTL}$ as shown by Lemmas \ref{cnf-open} and 
\ref{final} in Appendix \ref{app:tmtl}. Thus, $a \until_{(1,2), \#d=3 \wedge \#{\cnt^{<2}_{(0,1)}}\leq 5}c$ is expressible in $\mathsf{CTMTL}$. 
The \emph{nesting depth} of a $\mathsf{CTMTL}$ formula is the maximum nesting of $\mathsf{C, UT}$ operators. Formally, 
\begin{itemize}
\item  $depth(\varphi_1 \until_{I,\#{\varphi_3} \sim n} \varphi_2) = max(depth(\varphi_1),depth(\varphi_2),depth(\varphi_3) +1)$,
\item  $depth(C^{\ge n}_I \varphi) = depth(\varphi)+1$, $depth(\varphi \wedge \psi)=max(depth(\varphi), depth(\psi))$,
\item $depth(\neg \varphi)=depth(\varphi)$ and $depth(a) = 0$ for any $a \in \Sigma$.
\end{itemize}
\noindent For example, 
$depth(a \until_{[0,2], \eta} \cnt^{\geq 1}b)$ with 
$\eta=\textcolor{blue}{\#[}{a \until_{(0,1),\textcolor{red}{\#[}\cnt^{=2}_{(0,1)}a \wedge \fut_{(0,1),\textcolor{magenta}{\#}d = 2}\textcolor{red}{]}\geq 1}}c\textcolor{blue}{]} < 7$
 is 3. We obtain the following natural fragments of $\mathsf{CTMTL}$ as follows:
We denote by $\mathsf{CMTL}$, the fragment of $\mathsf{CTMTL}$ obtained by using the 
 $\mathsf{C}$ modality and the $\until_I$ modality.  Further, $\mathsf{C^0MTL}$
denotes the subclass of $\mathsf{CMTL}$ where the interval $I$ 
in $\cnt^{\sim n}_I\varphi$ is of the form 
$I=\langle 0,b \rangle$. When the interval is of the form $I=\langle 0,1\rangle$, then we denote 
the class by $\mathsf{C^{(0,1)}MTL}$. 
Note that $\mathsf{C^{(0,1)}MTL}$ is the class which allows 
counting in the next one unit of time.
  This kind of counting (unit counting in future and past) was introduced and studied in \cite{rabin} 
 in the continuous semantics.  $\mathsf{C^{(0,1)}MTL}$ is the pointwise counterpart 
 of this logic, with only future operators.
Clearly, $\mathsf{C^{(0,1)}MTL} \subseteq \mathsf{C^{0}MTL} \subseteq \mathsf{CMTL} \subseteq \mathsf{CTMTL}$.
 Restricting $\mathsf{CTMTL}$ to the $\mathsf{UT}$ modality, we obtain the fragment 
$\mathsf{TMTL}$. Restricting the $\mathsf{C}$ modality to 
$\mathsf{C^{(0,1)}}$ or $\mathsf{C^0}$ and also allowing the $\mathsf{UT}$ modality, one gets 
the fragments $\mathsf{C^{(0,1)}TMTL}$ and $\mathsf{C^{0}TMTL}$ respectively.   
If we disallow the $\mathsf{C}$ modality, restrict 
the intervals $I$ appearing in the formulae to non-punctual intervals of the form 
$\langle a, b \rangle$ ($a \neq b$), and restrict threshold formulae $\eta$ to be of the form 
$\#{true} \geq 0$, then we obtain $\mathsf{MITL}$.

\section{Expressiveness Hierarchy in the Counting Zoo}
 \label{sec:express}
 In this section, we study the expressiveness and hierarchy of the logics introduced in section \ref{zoo}. The main results of this section are the following:
 \begin{theorem}
 \label{thm:hier}
 $\mathsf{MTL} \subset \mathsf{C^{(0,1)}MTL} \subset \mathsf{C^0MTL} \subset \mathsf{TMTL} =\mathsf{C^{0}TMTL} \subset \mathsf{CTMTL}$. 
  Moreover,  $\mathsf{CMTL}$ and $\mathsf{TMTL}$ are incomparable, and  $\mathsf{C^0MTL}\subset \mathsf{CMTL}$.
 \end{theorem}
  While Theorem \ref{thm:hier} shows that there is an expressiveness  gap between classical $\mathsf{MTL}$ and $\mathsf{CTMTL}$, we show later that both these logics are   equisatisfiable. Given $\varphi \in  \mathsf{CTMTL}$, we can construct a formula $\psi \in  \mathsf{MTL}$ such that $\varphi$ is satisfiable iff $\psi$ is. 
 Note that our notion of  equisatisfiability is a special one {\it modulo temporal projections}. 
 If $\varphi$ is over an alphabet $\Sigma$, $\psi$ is constructed 
 over a suitable alphabet $\Sigma' \supseteq \Sigma$ 
 such that  $L(\psi)$, when projected over to $\Sigma$ gives  $L(\varphi)$.  

 \begin{theorem}
 \label{thm:es1}
Satisfiability Checking of $\mathsf{CTMTL}$ is decidable over finite timed words. 
 \end{theorem}

 
The rest of this paper is devoted to the proofs of Theorems \ref{thm:hier} and \ref{thm:es1}.
We establish Theorem \ref{thm:hier} through Lemmas \ref{tmtlless} to \ref{cmtlless}. 
To prove the separation between two logics,  we define  model-theoretic games. 

\subsection{Model-Theoretic Games}
Our games are inspired from the standard model-theoretic games \cite{straubing}, \cite{concur11}. 
 The 
$\mathsf{MTL}$ games introduced in \cite{concur11}  can be found in Appendix \ref{app:mtlgame}. 
 We introduce  $\mathsf{CTMTL}$ games. 
  \paragraph*{$\mathsf{CTMTL}$ Games}
Let $(\rho_1, \rho_2)$ be a pair of  timed words.
We define a  $r$-round $k$-counting pebble $I_{\nu}$ game on  $(\rho_1, \rho_2)$.
The game is played on  $(\rho_1, \rho_2)$ by two players, the $\spl$ and 
the $\dpl$. The $\spl$ will try to show that $\rho_1$ and $\rho_2$ are $\{r,k\}$-distinguishable 
by some formula in $\mathsf{CTMTL}$\footnote{$\rho_1, \rho_2$ are $\{r,k\}$-distinguishable 
iff there exists a $\mathsf{CTMTL}$ formula $\varphi$ having $depth(\varphi) \leq r$ with
 max counting constant $\leq k$ in any threshold formula $\eta$ or $\mathsf{C}$ modality in $\varphi$ such that 
$\rho_1 \models \varphi$ and $\rho_2 \nvDash \varphi$ or vice-versa.} while the $\dpl$ will try to show that $\rho_1, \rho_2$ 
are $\{r,k\}$-indistinguishable  in $\mathsf{TMTL}$. Each player has $r$ rounds and has access to a finite set of $\leq k$ pebbles 
from a box of pebbles $\mathcal{P}$
in each round of the game.  
Let $I_{\nu}$ be the set of permissible intervals allowed in the game. 

A configuration of the game at the start of a round $p$ is a pair of points $(i_p, j_p)$ where $i_p \in dom(\rho_1)$ and $j_p \in dom(\rho_2)$. 
A configuration is called partially isomorphic, denoted $isop(i_p ,j_p)$ iff $\sigma_{i_p} = \sigma_{j_p}$.
Exactly one of the $\spl$ or the $\dpl$ eventually wins the game.  
The initial configuration is $(i_1, j_1)$, the starting positions of both the words, before the first round. 
A 0-round game is won by the $\dpl$ iff $isop( i_1, j_1)$.  The $r$ round game is played by first playing one round from the starting position. Either the $\spl$ wins the round, and the game is terminated or the $\dpl$ wins the round, and now the second round is played from this new configuration and so on. The $\dpl$ wins the game only if he wins all the rounds. 
 The following are the rules 
 of the game in any round. Assume that the current configuration is $(i_p, j_p)$.
	\begin{itemize}
		\item If $isop(i_p,j_p)$ is not true, then $\spl$ wins the game, and 
		the game is terminated. Otherwise, the game continues as follows:
\item The $\spl$ chooses one of the words by choosing $\rho_x,  x \in \{1,2\}$. 
		$\dpl$ has to play on the other word $\rho_y$, $x \neq y$.  Then $\spl$ plays either a  $\until_{I, \eta}$ round, by choosing an interval $I\in I_{\nu}$, and a number $c \leq k$ of counting pebbles to be used,  or a $\cnt^{\sim c}_I$ round 
		by choosing an interval $I\in I_{\nu}$ and a number $c \leq k$ of counting pebbles to be used. 
		The number $c$ is obtained from ${\eta=\#{\varphi} \geq c}$ or ${\eta=\neg(\#{\varphi} \geq c})$.\\
\noindent{$\until_{I, \eta}$ \emph{round}}: 		
		 Given the current configuration as $(i_p,j_p)$ with $isop(i_p,j_p)$, then  
		\begin{itemize}
			\item $\spl$ chooses a position $i'_p \in dom(\rho_x)$ such that $i_p < i'_p$  and $(t_{i'_p} - t_{i_p}) 
			\in I$.
			\item The $\dpl$ responds by choosing $j'_p \in dom(\rho_y)$ in the other word such that $j_p < j'_p$ 
			and $(t_{j'_p} - t_{j_p}) \in I$.   
			If the $\dpl$ cannot find such a position, the
			$\spl$ wins the round and the game. Otherwise, the game continues and $\spl$ chooses one of the following three options. 
			\item $\fut$ Part: The round ends with the configuration $(i_{p+1}, j_{p+1})=(i'_p, j'_p)$.
		 \item $\until$  Part:
				 $\spl$ chooses a position $j''_p$ in $\rho_y$ such that $j_p < j''_p < j'_p$.
				  The $\dpl$
			responds by choosing  a position $i''_p$ in $\rho_x$ such that $i_p < i''_p < i'_p$.
			 The round ends with the configuration $(i_{p+1}, j_{p+1})=(i''_p,j''_p)$.  
			If  $\dpl$ cannot choose an $i''_p$, the game ends with $\spl$'s win.
	\item Counting Part : 
			First, $\spl$ chooses one of the two words to play in the counting part.  
			In his chosen word, $\spl$ keeps $c \le k$ pebbles from $\mathcal{P}$
			at $c$ distinct positions between the points $j_p$ and $j'_p$ (or $i_p$ and $i'_p$ depending on the choice of the word).  In response, the $\dpl$ also keeps $c$ pebbles from $\mathcal{P}$ at $c$ distinct positions between the points $i_p$ and $i'_p$  
			  (or $j_p$ and $j'_p$) in his word. 
			  $\spl$ then chooses a pebbled position say 
			  $i''_p$  (note that $i_p < i''_p < i'_p$) 
			  in the $\dpl$'s word.   In response, $\dpl$ chooses a pebbled position 
			  $j''_p$   (note that $j_p < j''_p < j'_p$) 
			  in the $\spl$'s word, and the game continues from the configuration $(i_{p+1}, j_{p+1})=(i''_p, j''_p)$. 
			  At the end of the round, the pebbles are returned to the box of pebbles $\mathcal{P}$.
			  \end{itemize}
\noindent{$\cnt_I^{\sim c}$ \emph{round}}: Given the current configuration as $(i_p,j_p)$ with $isop(i_p,j_p)$, 
 $\spl$ chooses  an interval $I\in I_{\nu}$ as well as a number $c \leq k$.  $\spl$ then chooses one of the words to play (say $\rho_1$). 
From $i_p$, $\spl$ places $c$ pebbles from $\mathcal{P}$ in the points lying in the interval $t_{i_p}+I$. In response, 
  $\dpl$ also places $c$ pebbles from $\mathcal{P}$ in the points lying in $t_{j_p}+I$. 
  $\spl$ now picks a pebbled position $j'_p$ in the word $\rho_2$, while $\dpl$ picks a pebbled position 
 $i'_p$ in the $\spl$'s word. The round ends
 with the configuration $(i'_p,j'_p)$. At the end of the round, the pebbles are returned to the box of pebbles $\mathcal{P}$.  \\ 
			  \noindent{\it Intuition on Pebbling}:
			  To give some intuition behind the pebbling, consider $\#{\varphi} \geq c$ or $\cnt^{\geq c}_I\varphi$.
			   The idea behind $\spl$ keeping $c$ pebbles on his word in the chosen interval $I$ is to say that these are the $c$ points where $\varphi$ evaluates to true. $\dpl$ is expected to find $c$ such points in his word. If $\spl$ suspects that in the $\dpl$'s 
			  word,  there are $< c$ positions in $I$ where $\varphi$ holds good, he picks up the appropriate pebble at the position where $\varphi$ fails. However, any pebbled position in $\spl$'s word will satisfy $\varphi$.  In this case, $\dpl$ loses. Similarly, if we have $\neg(\#{\varphi} \geq c)$, 
		or $\cnt^{<c}_I\varphi$, 	  
			  then $\spl$ chooses the word (say $\rho_1$) on which $\varphi$ evaluates to true $\geq c$ times. Then $\dpl$ is on $\rho_2$.  
			  The idea is for $\spl$ to find if there exist $c$ or more positions in the interval $I$ in $\rho_1$ where $\varphi$ 
			  holds good, and if so, pebble those points. This is based on $\spl$'s suspicion that there are atleast 
			  $c$ positions in $I$ where $\varphi$ evaluates to true, violating the formula. 
			  In response, $\dpl$ does the same 
			  on $\rho_2$. $\spl$ will now pick any one of the $c$ pebbles from $\rho_2$ and check for $\neg \varphi$. This is again based on $\spl$'s belief that whichever $c$ points $\dpl$ pebbles in $\rho_2$, $\neg \varphi$ will evaluate to true in atleast one of them. If $\varphi$ holds at all the $c$ points in $\rho_1$,  then $\dpl$ will lose on picking any pebble from $\rho_1$. 
			\item We can restrict various moves according to the modalities provided by the logic. For example, in a  $\mathsf{TMTL}[\fut_I]$ game, the possible rounds are $\fut_I$ and $\fut_{I, \eta}$. A $\mathsf{CMITL}$ game has only $\until_I, \cnt^{\geq n}_I$ rounds, with $I_{\nu}$ containing only non-punctual intervals.
	\end{itemize}
 
\noindent \textbf{Game equivalence:} $(\rho_1, i_1) \approx_{r,k,I_{\nu}} (\rho_2, j_1)$ iff for every $r$-round, $k$-counting pebble $\mathsf{CTMTL}$ game over the
	words $\rho_1, \rho_2$ starting from the configuration $(i_1, j_1)$, the $\dpl$ always has a winning strategy.\\
\textbf{Formula equivalence:} $(\rho_1, i_1) \equiv^{\mathsf{CTMTL}}_{r,k,I_{\nu}} (\rho_2, j_1)$ iff for every $\mathsf{CTMTL}$ formula $\varphi$ of depth $\le r$ having max counting constant $\leq k$ in the $\mathsf{C, UT}$  modalities,
	 $\rho_1, i_1 \models \varphi \iff \rho_2, j_1\models \varphi$. The proof of Theorem \ref{ctmtl-game} can be found in Appendix \ref{app:tmtlgame}.
\begin{theorem}
	\label{ctmtl-game}
	$(\rho_1, i_1) \approx_{r,k,I_{\nu}} (\rho_2,  j_1)\ iff\ (\rho_1, i_1) \equiv^{\mathsf{CTMTL}}_{r,k,I_\nu} (\rho_2,  j_1)$
\end{theorem}
 We now use these games to show the separation between various logics. 
For brevity, from here on, we omit $I_{\nu}$ from the notations
$\equiv^{\mathsf{CTMTL}}_{r,k,I_{\nu}}$, 
 $\equiv^{\mathsf{CMTL}}_{r,k,I_{\nu}}$, $\equiv^{\mathsf{TMTL}}_{r,k,I_{\nu}}$ and  $\equiv^{\mathsf{MTL}}_{r,I_{\nu}}$.  
\begin{lemma}
\label{tmtlless}
$\mathsf{CMTL} - \mathsf{TMTL} \neq \emptyset$
\end{lemma}
\begin{proof}
Consider the formula $\varphi=\cnt_{(1,2)}^{\geq 2}a \in \mathsf{CMTL}$.  
We show that for any choice of $n$ rounds and $k$ pebbles, 
we can find two words $\rho_1, \rho_2$ such that 
 $\rho_1 \models \varphi, \rho_2 \nvDash \varphi$, but 
 $\rho_1 \equiv_{n,k}^{\mathsf{TMTL}} \rho_2$.
 Both $\rho_1, \rho_2$ are over $\Sigma=\{a\}$. 
  Let  $0< \delta<\epsilon < \frac{1}{10^{10nk}}$ and 
 $0< \kappa < \frac{\epsilon-\delta}{2nk}$.
 Let $l$ be the maximum constant in $\mathbb{N}$ appearing in the 
 permissible intervals $I_{\nu}$.  
 Consider the word $\rho_1$ with $nl(k+1)=K$ unit intervals,  with the following  time stamps as depicted pictorially (Figure \ref{mainexpress:fig1}) and in the table.

\flushleft{
\begin{figure}[t]
\begin{tikzpicture}

\foreach \x in {0.25}{
\draw  (\x + 0,0) -- (\x + 8,0);
\draw[dashed] (\x+8,0) -- (\x+10,0);
\draw (\x+10,0) -- (\x + 12,0);
\foreach \y in {0,1,2,3,4}
{
	
\draw (\x+\y*2,-0.25)--(\x+\y*2,0.25);

\node at (\x + \y*2, -0.35) {\tiny \y};

}
\foreach \y in {5,6}
{
	
	\draw (\x+\y*2,-0.25)--(\x+\y*2,0.25);

}
	\node at (\x + 10, -0.35) {\tiny $K$};
		\node at (\x + 12, -0.35) {\tiny $K+1$};



\foreach \y in {0.85,1.1,1.8,3.75,3.9,4.1,4.75,6.7,6.9,7.55,10.5,10.75,11.45}
{\node[fill = magenta,draw = blue,circle,inner sep=1pt,label=below:{}] at (\x+\y,0){};
}
\foreach \y in {1.2,1.25,..., 1.65}
{\node[fill = blue,draw = blue,circle,inner sep=0.5pt,label=below:{}] at (\x+\y,0){};
}

\foreach \y in {4.2,4.25,...,4.65}
	{\node[fill = blue,draw = blue,circle,inner sep=0.25pt,label=below:{}] at (\x+\y,0){};
}
\foreach \y in {7,7.05,...,7.45}
{\node[fill = blue,draw = blue,circle,inner sep=0.25pt,label=below:{}] at (\x+\y,0){};
}
\foreach \y in {10.85,10.9,...,11.25}
{\node[fill = blue,draw = blue,circle,inner sep=0.25pt,label=below:{}] at (\x+\y,0){};
}

\node at (\x + 0.85, -0.25) 
{\tiny {$x_1$} };
\node at (\x + 1.1, 0.15) 
{\tiny {$z_1$} };
\node at (\x + 1.8, -0.25) 
{\tiny {$y_1$} };
\node at (\x + 3.75, -0.25) 
{\tiny {$x_2$} };
\node at (\x + 3.9 , 0.15) 
{\tiny {$z_2$} };

\node at (\x + 4.1, 0.15) 
{\tiny {$e$} };

\node at (\x + 4.75, -0.25) 
{\tiny {$y_2$} };		
\node at (\x + 6.7, -0.25) 
{\tiny {$x_3$} };
\node at (\x + 6.9, 0.15) 
{\tiny {$z_3$} };
\node at (\x + 7.55, -0.25) 
{\tiny {$y_3$} };

%
%
%
%

%
\node at (\x,-0.5){};
			
		\node at (\x + 10.5 , -0.25){\tiny {$x_{K}$} };
		\node at (\x + 10.75, 0.15) 
		{\tiny {$z_{K}$} };
		\node at (\x + 11.5 , -0.25) 
		{\tiny {$y_{K}$} };
}

\end{tikzpicture}
}
\flushleft{
	\begin{tikzpicture}
	\foreach \h in {0.1}
	{
	\foreach \x in {0.25}{
		\draw (\x + 0,0) -- (\x + 8,0);
		\draw[dashed] (\x+8,0) -- (\x+10,0);
		\draw (\x+10,0) -- (\x + 12,0);
		\foreach \y in {0,1,2,3,4}
		{
			
			\draw (\x+\y*2,-0.25)--(\x+\y*2,0.25);
			
			\node at (\x + \y*2, -0.35) {\tiny \y};
			
		}
		\foreach \y in {5,6}
		{
			
			\draw (\x+\y*2,-0.25)--(\x+\y*2,0.25);

		}
		\node at (\x + 10, -0.35) {\tiny $K$};
		\node at (\x + 12, -0.35) {\tiny $K+1$};


			 {0.85,1.1,1.8,3.79,4,4.75,6.7,6.9,7.55,10.5,10.75,11.45}

			\node at (\x + 0.85 -\h, -0.25) 
				{\tiny {$x_1'$} };
			\node at (\x + 1.1 -\h, 0.15) 
			{\tiny {$z_1'$} };
			\node at (\x + 1.8 -\h, -0.25) 
			{\tiny {$y_1'$} };
			\node at (\x + 10.5, -0.25) 
				{\tiny {$x_{K}'$} };
				\node at (\x + 10.75, 0.21) 
				{\tiny {$z_{K}'$} };
				\node at (\x + 11.5 -\h, -0.25) 
				{\tiny {$y_{K}'$} };
			
			\node at (\x + 3.9 -\h, -0.25)  {\tiny {$x_2'$} };
			\node at (\x + 4.3 -\h, 0.21) 
			{\tiny {$z_2'$} };
			\node at (\x + 4.75 -\h, -0.25) 
			{\tiny {$y_2'$} };		
			\node at (\x + 6.7 -\h, -0.25) 
			{\tiny {$x_3'$} };
			\node at (\x + 6.9 -\h, 0.15) 
			{\tiny {$z_3'$} };
			\node at (\x + 7.55 -\h, -0.25) 
			{\tiny {$y_3'$} };	
						
		\foreach \y in {0.85,1.1,1.75,3.95,4.2,4.75,6.7,6.9,7.55,10.5,10.75,11.45}
		{\node[fill = magenta,draw = blue,circle,inner sep=1pt,label=below:{}] at (\x+\y - \h ,0){};
		}
		\foreach \y in {1.2,1.25,..., 1.65}
		{\node[fill = blue,draw = blue,circle,inner sep=0.5pt,label=below:{}] at (\x+\y - \h,0){};
		}
		
		\foreach \y in {4.3,4.35,...,4.65}
		{\node[fill = blue,draw = blue,circle,inner sep=0.25pt,label=below:{}] at (\x+\y-\h,0){};
		}
		\foreach \y in {7,7.05,...,7.45}
		{\node[fill = blue,draw = blue,circle,inner sep=0.25pt,label=below:{}] at (\x+\y-\h,0){};
		}
		\foreach \y in {10.95,11,...,11.45}
		{\node[fill = blue,draw = blue,circle,inner sep=0.25pt,label=below:{}] at (\x+\y-\h-\h,0){};
		}
		%
		%
		%
		%
		
		%
		}
}
	\end{tikzpicture}
\caption{Words showing $\mathsf{CMTL}-\mathsf{TMTL} \neq \emptyset$}
\label{mainexpress:fig1} 
\end{figure}
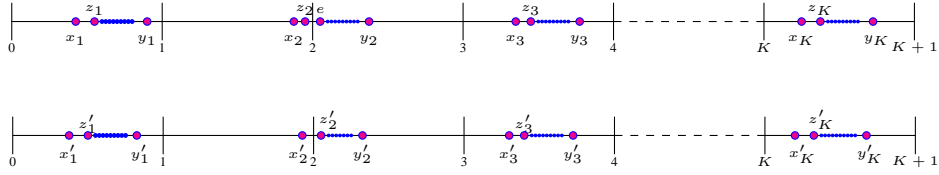
}
\vspace{-0.5cm}
\begin{table}
  \begin{tabular}{|c|c|c|}
    \hline
 Points in & $\rho_1$ & $\rho_2$\\
  \hline
  (0,1)& $x_1=0.5, z_1=0.6, y_1=0.8$ &   $x'_1=0.5, z'_1=0.6, y'_1=0.8$ \\
           & and $2nk$ points between $z_1, y_1$ & and $2nk$ points  between $z'_1, y'_1$ \\
           &that are $\kappa$ apart from each other & that are $\kappa$ apart from each other\\                                         
            \hline
  (1,2) &   $x_2=1.8-\epsilon, z_2=1.8+\epsilon$ &    $x'_2=1.8-\epsilon$\\
  \hline
  (2,3) &  $e=2.4+n\epsilon, y_2=2.7+n\epsilon$    & $z'_2=2.4+n\epsilon,   y'_2=2.7+n\epsilon$  \\
         &  and $2nk$ points between $e$ and $y_2$   & and $2nk$ points between $z'_2$ and $y'_2$ \\
         &that are $\kappa$ apart from each other & that are $\kappa$ apart from each other\\
         \hline
   $(i,i+1)$ & $x_i=i+0.4+(n-i)\epsilon$ &  $x'_i=i+0.4+(n-i)\epsilon$\\
$3 \leq i \leq K-1$  & $z_i=i+0.8+(n+i)\epsilon+\delta$ &  $z'_i=
i+0.8+(n+i)\epsilon+\delta$ \\
               & $y_i=i+0.8+(n+i+1)\epsilon$ and $2nk$points
   &  $y'_i=i+0.8+(n+i+1)\epsilon$ and $2nk$points\\
 &  between $z_i, y_i$ that are   $\kappa$ apart from each other & between $z_i, y_i$ that are   $\kappa$ apart from each other\\
\hline 
  \end{tabular}
\end{table} 
Thus, $\rho_1$ and $\rho_2$ differ only in the interval (1,2) : 
$\rho_1$ has two points in (1,2), while $\rho_2$ has only one. 
Thus, $\rho_1 \models \varphi, \rho_2 \nvDash \varphi$. 

Let $seg(i_p) \in \{0, 1, \dots, K\}$  denote the left endpoint  of the left closed, right open unit interval 
containing the point $i_p \in dom(\rho_1)$ or $dom(\rho_2)$. Our segments are [0,1), [1,2), $\dots$, $[K, K+1)$. 
    For instance,  if the configuration at the start of the $p$th round is  $(i_p, j_p)$ with time stamps (1.2, 3), then $seg(i_p)=1, seg(j_p)=3$. The following lemma says that in any round of the game, $\dpl$ can either achieve 
    the same segment in both the words, or ensure that the difference in the segments is atmost 1. 
    Moreover, by the choice of the words, there are sufficiently many segments 
    on the right of any configuration so that $\dpl$ can always duplicate $\spl$'s moves 
  for the remaining rounds, preserving the lag of one segment.   \\
    
\noindent{\it Copy-cat strategy} Consider the $p$th round of the game with configuration $(i_p, j_p)$. 
If $\dpl$ can ensure that $seg(i_{p+1}){-}seg(i_p)$=$seg(j_{p+1}){-}seg(j_p)$, then 
we say that $\dpl$ has adopted a {\it copy-cat} strategy in the $p$th round. 
We prove the following proposition to argue $\dpl$'s win.

\begin{proposition}
\label{game-details}
For an $n$ round $\mathsf{TMTL}$ game over the words $\rho_1, \rho_2$, 
the $\dpl$ always has a winning strategy such that for any $1 \leq p \leq n$, if 
$(i_p, j_p)$ is the initial configuration of the $p^{th}$ round, then 
$|seg(i_p)-seg(j_p)| \leq 1$. Moreover, when 
$|seg(i_p)-seg(j_p)| =1$, then there are atleast $(n-p)(l+1)$ segments to the right 
on each word after $p$ rounds, for all $1 \leq p \leq n$.
    \end{proposition}
\begin{proof}
The initial configuration has time stamps (0,0). We will play a $(n,k)$-$\mathsf{TMTL}$ game
on $\rho_1, \rho_2$. Assume that the $\spl$ chooses $\rho_1$ while the $\dpl$ chooses 
$\rho_2$. 
Since the interval [1,2] is the only one different in both the words, 
it is interesting to look at the moves where the $\spl$ 
chooses a point in interval (1,2). We consider the two situations
possible for $\spl$ to land up in a point in interval (1,2): 
he can enter interval (1,2) from some point in interval (0,1), 
or directly choose to enter interval (1,2) from the initial configuration with time stamps (0,0). \\
\noindent{\bf Situation 1}: 
Consider the case when from the starting configuration $(i_1, j_1)$ with time stamps (0,0), 
 $\spl$ chooses a $\until_{(1,2)\#a\sim c}$ move
 in $\rho_1$ and lands up at the point $x_2$ or $z_2$. 
 In response, $\dpl$ has to come at the point $x'_2$ 
 in $\rho'_2$.  
 If $(i'_1, j'_1)$ has time stamps $(x_2, x'_2)$
 and if $\spl$ chooses to pebble between 0 and $x_2$, then 
 $\dpl$ pebbles between 0 and $x'_2$; however, an identical configuration is obtained. 
 Note that if $\spl$ pebbles $\rho_2$, then $\dpl$
 has it easy, since he will pebble the same positions in $\rho_1$. 
  Let us hence consider obtaining the configuration 
 $(i'_1, j'_1)$ with time stamps $(z_2, x'_2)$, and let $\spl$ pebble $\rho_1$.
 $\spl$ can keep a maximum of $k$ pebbles 
 in the points $x_1, \dots, y_1, x_2$, while 
 $\dpl$ keeps the same number of  pebbles 
 on the points $x'_1, \dots, y'_1$. In this case, 
 $\spl$ has to a pick a pebbled position 
 from among $x'_1, \dots, y'_1$. 
 In response, $\dpl$ will pick the same position 
 from $\spl$'s word and achieve an identical configuration.  
 An interesting special case is  when $\spl$ keeps a single pebble at $x_2$ in $\rho_1$.
 In this case, $\dpl$'s best choice is to keep his pebble at $x'_1$, so that 
 the next configuration $(i_2, j_2)$ is one with time stamps  $(x_2,x'_1)$. 
  $x'_1$ and $x_2$ are {\it topologically similar} in the sense that 
  the distribution of points in subsequent segments 
  have some nice properties as given below.\\
\noindent \emph{Topological Similarity of Words}: Consider the $2nk+3$ points $x_j< $ $z_j <$ $p^1_j <$ $\dots <$ $p^{2nk}_j <$ $y_j$ in $\rho_1$, 
  and $x'_{j-1}$ $< z'_{j-1}$  $< q^1_{j-1} < \dots$ $< q^{2nk}_{j-1}$ $< y'_{j-1}$ in $\rho_2$, for 
  $j \in \{2,3, 4, \dots, K\}$. Define a function $f$ that maps points in $\rho_1$ to topologically similar points in $\rho_2$.
   $f:\{x_j, z_j, p^1_j, \dots, p^{2nk}_1, y_j\} \rightarrow  
  \{x'_{j-1}, z'_{j-1}, q^1_{j-1}, \dots, q^{2nk}_{j-1}, y'_{j-1}\}$ by $f(x_j)=x'_{j-1}, f(z_j)=z'_{j-1}, f(y_j)=y'_{j-1},$
  $f(p^i_j)=q^i_{j-1}$. Let $g=f^{-1}$. 
   \begin{enumerate}
 \item[(a)] The current configuration has timestamps $(x_2, x'_1)=(x_2, f(x_2))$. For $j \geq 2$, if $\spl$ chooses 
  to move to any $p \in \{z_j, y_j, x_{j+2}\}$ from $x_2$, then $\dpl$ can move to $f(p)$
  from $f(x_2)$ since, for any time interval $I$, 
  it can be seen that $p-x_2 \in I$ iff  $f(p)-f(x_1)\in I$. Moreover, if $\spl$ chooses 
  to move to  $x_3$ from $x_2$, then $\dpl$ can move to $z'_2$
  from $f(x_2)$ since, $x_3-x_2, z'_2-f(x_2) \in (0,1)$.
  \item[(b)] We can extend (a) above as follows: Let the current configuration have timestamps 
  $(p,f(p))$ or $(x_3, z'_2)$.  Then it can be seen that for any 
  $q \in \{x_j,y_j,z_j\}$ and interval $I$, $q-p \in I$ iff 
  $f(q)-f(p) \in I$, and $q-x_3 \in I$ iff $f(q)-z'_2 \in I$. 
    \end{enumerate}   
  The facts claimed in (a) and (b) are evident from the construction of the timed words. 
  They show that from a configuration $(i_p, j_p)$, 
  such that $seg(i_p)-seg(j_p) \leq 1$, 
    $\dpl$ can always achieve an intermediate configuration 
  $(i'_p, j'_p)$ in any $\until_{I, \#a \sim c}$ such that $seg(i'_p)-seg(j'_p) \leq 1$.
    If $\spl$ does not go for the until round or the counting round, then $(i_{p+1},j_{p+1})=(i'_p,j'_p)$.  If $\spl$ pebbles the points between $i_p$ and $i'_p$ (or $j_p$ and $j'_p$), then 
      $\dpl$ can always ensure that he pebbles points $f(P)$ in $\rho_2$ whenever  $\spl$ pebbles a set of points $P$ in $\rho_1$. 
  As a result, if $\spl$ chooses a point $q=f(i) \in f(P)$ in $\rho_2$, then $\dpl$ can choose
  the point $g(q)=i\in P$ achieving the configuration $(i_{p+1},j_{p+1})=(g(q),q)=(i,f(i))$. 
  By definition of $f,g$, we have $i_{p+1}-j_{p+1} \leq 1$.
  Note that $\dpl$ can also achieve an identical configuration 
  if $\spl$ moves ahead by several segments from $i_p$ (thus, $i'_p >>i_p$), and pebbles 
  a set of points that are also present between $j_p$ and $j'_p$. \\
\noindent{\bf Situation 2}: 
Starting from $(i_1, j_1)$ with time stamps (0,0), 
if the $\spl$ chooses a $\until_{(0,1),\#a \sim c}$ move and 
lands up at some point between $x_1$ and $y_1$, $\dpl$ will play copy-cat
and achieve an identical configuration.  
 Consider the case when 
$\spl$ lands up at $y_1$\footnote{The argument when $\spl$ lands up at $x_1$ or a point in between $x_1, y_1$ is exactly the same}. In response,  $\dpl$ moves to   
$y'_1$. From configuration $(i_2, j_2)$ with time stamps $(y_1, y'_1)$, 
consider the case when $\spl$ initiates a $\until_{(1,2),\#a \sim c}$ and moves 
 to $z_2=1.8+\epsilon <2$. In response, $\dpl$ 
 moves to the point $z'_2=2.1>2$. 
 A pebble is kept at the inbetween 
positions $x_2, x'_2$ respectively in $\rho_1, \rho_2$. 
When $\spl$  picks the pebble in 
$\dpl$'s word, then we obtain the configuration $(i_3, j_3)$ with time stamps  
$(x_2, x'_2)$.  If $\spl$ does not get into the counting part/until part, the 
configuration obtained has time stamps  $(z_2, z'_2)$, with the lag 
of one segment ($seg(i_3)=1$, $seg(j_3)=2$, $seg(j_3)$-$seg(i_3)$=1). 
We show in Appendix \ref{app:situation2}
that from  $(i_3,j_3)$ 
with time stamps either $(x_2, x'_2)$ or $(z_2, z'_2)$, $\dpl$ can either achieve an identical configuration, 
or achieve a configuration with a lag of one segment.
   \end{proof}
From  situations  (1), (2) in Proposition \ref{game-details}, we know that either $\dpl$ achieves an identical configuration, in which case there is no segment lag, or there is a lag of atmost one segment. The length of the words are $lnk+nl=K$. 
If $\spl$ always chooses bounded intervals (of length $\leq l$), 
then $\dpl$ respects his segment lag of 1, and the maximum number of segments that can be explored in either word 
is atmost $nl < K$. In this case, after $p$ rounds, there are atleast
$K-pl$ $\geq$ $nlk+nl-pl$ $\geq$ $(n-p)(l+1)$ segments to the right of $\rho_1$ and $K-pl+1$ segments to the right of $\rho_2$.
If $\spl$ chooses an unbounded interval in any round, then $\dpl$ 
can either enforce an identical configuration in both situations 1 and 2,   
or obtain one of the configurations with time stamps 
$(p, f(p))$, $f(p) \neq x'_2$, or $(z_2,x'_2)$ or $(x_2, x'_2)$, from where 
it is known that $\dpl$ wins.
\end{proof}
  \begin{lemma}
\label{lem1}
$\mathsf{MTL} \subset \mathsf{C^{(0,1)}MTL} \subset  \mathsf{C^0MTL}$
\end{lemma}
\begin{proof}
We show that the formula $\varphi=\cnt^{=2}_{(0,1)}a \in \mathsf{C^{(0,1)}MTL}$
cannot be expressed in $\mathsf{MTL}$. Likewise, the formula 
$\varphi=\cnt_{(0,2)}^{\geq 2} a \in \mathsf{C^0MTL}$ cannot be expressed in 
$\mathsf{C^{(0,1)}MTL}$. A detailed 
proof of these are given by Propositions \ref{app:exp1} and 
\ref{app:exp2} in Appendix \ref{app:lem1}.
\end{proof}
\begin{lemma}
\label{c0intmtl}
\begin{itemize}
\item[(i)] $\mathsf{C^0MTL} \subset \mathsf{TMTL}=\mathsf{C^0TMTL}=\mathsf{C^{(0,1)}TMTL}$ and \item[(ii)]$\mathsf{C^0MTL} \subset \mathsf{CMTL}$.
\end{itemize}
\end{lemma}
\begin{proof}
(i) The first containment as well as the last two equalities follows from the  fact that the counting modality $\cnt_{\langle 0,j\rangle}^{\geq n}\varphi$  of $\mathsf{C^0MTL}$
can be written in $\mathsf{TMTL}$ as $\fut_{\langle 0, j \rangle, \#{\varphi} \geq n} true$.
The strict containment of $\mathsf{C^0MTL}$ then follows from Lemma \ref{cmtlless}.  
(ii) We know that $\mathsf{C^0MTL} \subseteq \mathsf{CMTL}$. This along with 
(i) and  Lemma \ref{tmtlless} gives the strict containment.  
\end{proof}
\begin{lemma}
\label{cmtlless}
$\mathsf{TMTL} - \mathsf{CMTL} \neq \emptyset$
\end{lemma}
\begin{proof}
 Consider the formula $\varphi=\fut_{(0,1), \#a \geq 3}b \in \mathsf{TMTL}$.
We show that for any choice of $n$ rounds and $k$ pebbles, 
we can find two words $\rho_1, \rho_2$ such that 
 $\rho_2 \models \varphi, \rho_1 \nvDash \varphi$, but 
 $\rho_1 \equiv_{n,k}^{\mathsf{CMTL}} \rho_2$.  The words can be seen  in Figure \ref{main:fig2} and the details in 
  Appendix \ref{app:cmtl-less}.
\flushleft{
\begin{figure}[th]
\begin{tikzpicture}

\foreach \x in {0.25}{
\draw  (\x + 0,0) -- (\x + 6,0);
\draw[dashed] (\x+6,0) -- (\x+8,0);
\draw (\x+8,0) -- (\x + 11,0);
\foreach \y in {0,1,2,3}
{
	
\draw (\x+\y*2,-0.25)--(\x+\y*2,0.25);

\node at (\x + \y*2, -0.3) {\tiny \y};

}
\foreach \y in {4,5}
{
	
	\draw (\x+\y*2,-0.25)--(\x+\y*2,0.25);

}
	\node at (\x + 8, -.3) {\tiny $K-1$};
		\node at (\x + 10, -.3) {\tiny $K$};

\foreach \u in {0,2,4,8}
{	

	\foreach \y in {0.55,1.15,1.75}
	\node[fill = red,draw = black,rectangle,inner sep=1.25pt,label=below:{}] at (\x+\y+\u,0){};

	\foreach \z in {0.1,0.7,1.3}
{	\foreach \y in {0.08,0.11,...,0.3}
	{\draw[blue] (\x+\y+\z+\u,-0.15)--(\x+\y+\z+\u,0.15);
	} 

}
}
}
\end{tikzpicture}
}

	\flushleft{
		\begin{tikzpicture}
	
		\foreach \x in {0.25}{
			\draw  (\x + 0,0) -- (\x + 6,0);
			\draw[dashed] (\x+6,0) -- (\x+8,0);
			\draw (\x+8,0) -- (\x + 11,0);
			\foreach \y in {0,1,2,3}
			{
				
				\draw (\x+\y*2,-0.25)--(\x+\y*2,0.25);
				
				\node at (\x + \y*2, -0.3) {\tiny \y};
				
			}
			\foreach \y in {4,5}
			{
				
				\draw (\x+\y*2,-0.25)--(\x+\y*2,0.25);

			}
			\node at (\x +8, -0.3) {\tiny $K-1$};
			\node at (\x + 10, -0.3) {\tiny $K$};

			\foreach \u in {0,4}
			{	
				
				\foreach \y in {0.53,1.04,1.54}
				\node[fill = red,draw = black,rectangle,inner sep=1.25pt,label=below:{}] at (\x+\y+\u,0){};
				
				\foreach \z in {0.08,0.61,1.12,1.65}
				{	\foreach \y in {0.07,0.11,...,0.3}
					{\draw[blue] (\x+\y+\z+\u,-0.15)--(\x+\y+\z+\u,0.15);
					} 
					
				}
			}
			\foreach \u in {8}
			{	
				
				\foreach \y in {0.53,1.04,1.54}
				\node[fill = red,draw = black,rectangle,inner sep=1.25pt,label=below:{}] at (\x+\y+\u,0){};
				
				\foreach \z in {0.08,0.61,1.12}
				{	\foreach \y in {0.07,0.11,...,0.3}
					{\draw[blue] (\x+\y+\z+\u,-0.15)--(\x+\y+\z+\u,0.15);
					} 
					
				}
			}

		\foreach \u in {1.9}
		{	
			
			\foreach \y in {0.55,1.05,1.55}
			\node[fill = red,draw = black,rectangle,inner sep=1.1pt,label=below:{}] at (\x+\y+\u,0){};
			
			\foreach \z in {0.11,0.62,1.11,1.65}
			{	\foreach \y in {0.08,0.11,...,0.3}
				{\draw[blue] (\x+\y+\z+\u,-0.15)--(\x+\y+\z+\u,0.15);
				} 
				
			}
		}
		
	}
	\end{tikzpicture}
\caption{The red square represents $a$, the bunch of blue lines represents a bunch of $b$'s. 
There are 3 $a$'s in each unit interval of both $\rho_1$ and $\rho_2$. The difference is that 
$\rho_1$ has 3 blocks of $b$'s in each unit interval, while $\rho_2$ has 4 blocks of $b$'s in each unit interval except the last.  
Clearly, $\rho_2 \models \varphi, \rho_1 \nvDash \varphi$.}
		\label{main:fig2}
				\end{figure}
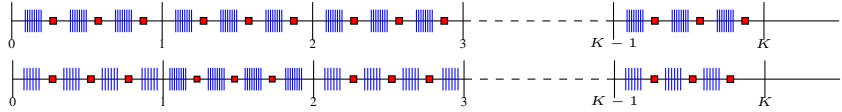
				}

  \end{proof}


 \section{Satisfiability Checking of Counting Logics}
\label{equis}
In this section, we  show that $\mathsf{CTMTL}$  has a decidable satisfiability checking. 
For this, given a formula in $\mathsf{CTMTL}$ 
 we synthesize an equisatisfiable 
formula in $\mtl$, and use the decidability of $\mtl$.
We start discussing some preliminaries.
Let $\Sigma,X$ be finite sets of propositions such that $\Sigma \cap X=\emptyset$.
\begin{enumerate}
\item {\it$(\Sigma,X)$-simple extensions}.  A $(\Sigma,X)$-simple extension is a timed word $\rho' = (\sigma',\tau')$ over $X \cup \Sigma$ such that at any point $i \in dom(\rho')$, $\sigma'_i \cap \Sigma \ne \emptyset$. 
For $\Sigma=\{a,b\}, X=\{c,d\}$,  
$(\{a\},0.2)(\{b,c,d\},0.3)(\{b,d\},1.1)$ is a  $(\Sigma,X)$-simple extension.  However, 
 $(\{a\},0.2)(\{c,d\},0.3)(\{b,d\},1.1)$ is not. 
\item {\it Simple Projections}. 
Consider a $(\Sigma,X)$-simple extension $\rho$. 
We define the {\it simple projection} of $\rho$ with respect to $X$, denoted $\rho \setminus X$ as the word obtained by erasing the symbols of $X$ from each $\sigma_i$.  Note that $dom(\rho)=dom(\rho \setminus X)$.
For example, if $\Sigma=\{a,c\}$, $X=\{b\}$, and 
$\rho=(\{a,b,c\},0.2)(\{b,c\},1)(\{c\},1.3)$, then 
$\rho \setminus X=(\{a,c\},0.2)(\{c\},1)(\{c\},1.3)$. 
$\rho \setminus X$ is thus, a timed word over $\Sigma$.
If the underlying word $\rho$ is {\it not} a $(\Sigma,X)$-simple extension, then 
$\rho \setminus X$
 is {\it undefined}.
\item {\it$(\Sigma,X)$-oversampled behaviours}.  
A $(\Sigma,X)$-oversampled behaviour is a timed word $\rho'=(\sigma', \tau')$ over $X \cup \Sigma$, such that 
$\sigma'_1 \cap \Sigma \neq \emptyset$ and  $\sigma'_{|dom(\rho')|} \cap \Sigma \neq \emptyset$. Oversampled behaviours are more general than simple extensions since they allow occurrences of new points in between the first and the last position. These new points are called {\it oversampled points}. All other points are called {\it action points}.
For $\Sigma=\{a,b\}, X=\{c,d\}$,  
$(\{a\},0.2)(\{c,d\},0.3)(\{a,b\},0.7)(\{b,d\},1.1)$ 
is a  $(\Sigma,X)$-oversampled behaviour, while 
$(\{a\},0.2)(\{c,d\},0.3)(\{c\},1.1)$ is not.  
\item {\it Oversampled Projections}.       
Given a $(\Sigma,X)$-oversampled behaviour $\rho'=(\sigma',\tau')$, 
the oversampled projection of $\rho'$ with respect to $\Sigma$, 
denoted $\rho' \downarrow X$ is defined as the timed word obtained 
by deleting the oversampled points, 
and then 
erasing the symbols of $X$  from the action points. 
 $\rho$=$\rho'\downarrow X$ is a timed word over $\Sigma$.
\end{enumerate}
A {\it temporal projection} is either a simple projection or an oversampled projection.  
 We now define \emph{equisatisfiability modulo temporal projections}. 
Given $\mathsf{MTL}$ formulae $\psi$ and $\phi$, we say that $\phi$ is equisatisfiable to $\psi$ 
{\it modulo temporal projections} iff there exist disjoint sets $X, \Sigma$ such that 
(1) $\phi$ is  over $\Sigma$, and $\psi$ over $\Sigma \cup X$,
(2) For any timed word $\rho$ over $\Sigma$ such that  $\rho \models \phi$,  there exists  a timed word $\rho'$ such that $\rho' \models \psi$, and $\rho$ is a temporal projection of $\rho'$ with respect to $X$,
(3) For any behaviour $\rho'$ over $\Sigma \cup X$,   if $\rho' \models \psi$ then  the temporal projection $\rho$ of $\rho'$ with respect to $X$  is well defined and $\rho\models \phi$.\\
\noindent If the temporal projection used above is a simple projection, we call it {\it equisatisfiability modulo simple projections} and denote it by 
$\phi=\exists X. \psi$. If the projection in the above definition is an oversampled projection, then it is called {\it equisatisfiability modulo oversampled projections} and is denoted $\phi\equiv\exists  \downarrow X. \psi$. Equisatisfiability modulo simple projections are studied extensively \cite{raskin-thesis}, \cite{deepak08}, \cite {kini08}.
 It can be seen that if $\varphi_1=\exists X_1. \psi_1$ and $\varphi_2=\exists X_2. \psi_2$, with $X_1$, $X_2$ disjoint, then 
 $\varphi_1\wedge \varphi_2=\exists (X_1 \cup X_2). (\psi_1\wedge \psi_2)$ \cite{arxiv-time14}.\\
As in the case of simple projections, equisatisfiability modulo oversampled projections are
	also closed under conjunctions when one considers the relativized formulae. 
		For example, consider a formula $\varphi = \Box_{(0,1)}a$ over $\Sigma =
		 \{a,d\}$. Let $\psi_1 = \Box_{(0,1)}(a \vee b) \wedge \fut_{(0,1)}(b \wedge \neg a)$ be a formula over the extended alphabet $\{a,b,d\} $ and $\psi_2 = \Box(c \leftrightarrow \Box_{(0,1)}a)\wedge  c $ over the extended alphabet $\{a,c,d\}$. Note that $\varphi = \exists \downarrow \{b\}. \psi_1$ and $\varphi = \exists \downarrow \{c\}. \psi_2$ but $\varphi \wedge \varphi \ne \exists \downarrow \{b,c\}. (\psi_1 \wedge \psi_2)$ as the left hand side evaluates to  $\varphi$ which is satisfiable while the right hand side is unsatisfiable. This is due to the presence of a \emph{non-action} point where only $b$ holds. But this can easily be fixed by relativizing all the formulae over their respective action points. $\psi_1$ is relativized as $\lambda_1 = \Box_{(0,1)}(act_1 \rightarrow (a \vee b)) \wedge \fut_{(0,1)}(act_1 \wedge b \wedge \neg a) $ and $\psi_2$ is relativized as $\lambda_2 = \Box(act_2 \rightarrow (c \leftrightarrow \Box_{(0,1)}(act_2 \rightarrow a)))\wedge  act _2 \wedge c$ where $act_1 = b \vee d \vee a$ and $act_2 = a \vee c \vee d$. Now, $\varphi \wedge \varphi = \exists \downarrow \{b,c\}. (\lambda_1 \wedge \lambda_2)$.
 	The relativized forms of $\psi_1, \psi_2$ are called their \emph{Oversampled Normal Forms} with respect to $\Sigma$ and denoted
	$ONF_{\Sigma}(\psi_1)$ and $ONF_{\Sigma}(\psi_2)$.
	Then it can be seen that $\varphi_1 \wedge \varphi_2=\exists \downarrow\{b,d\}. [ONF_{\Sigma}(\psi_1) \wedge ONF_{\Sigma}(\psi_2)]$, 
	and $\varphi_1=\exists \downarrow \{b\}. ONF_{\Sigma}(\psi_1)$, 
	$\varphi_2=\exists \downarrow \{d\}. ONF_{\Sigma}(\psi_2)$. The formal definition of $ONF_{\Sigma}(\varphi)$ 
	for a formula $\varphi$ over $\Sigma \cup X$ can be found in Appendix \ref{app:onf}.
	Equisatisfiability modulo oversampled projections were first studied in \cite{time14} to eliminate non-punctual past from $\mathsf{MTL}$ over timed words. 	We use equisatifiability modulo simple  projections to eliminate the $\mathsf{C}$ modality  and oversampled projections to eliminate the $\mathsf{UT}$ modality from $\mathsf{CTMTL}$.

 \subsection*{Elimination of Counting Modalities from $\mathsf{CTMTL}$}
\label{celim}
In this section, we show how to eliminate the counting constraints 
from  $\mathsf{CTMTL}$ 
over strictly monotonic timed words. 
This can be extended to weakly monotonic timed words.

Given any $\mathsf{CTMTL}$ formula $\varphi$ over $\Sigma$, we 
 ``flatten" the  $\mathsf{C, UT}$ modalities 
of $\varphi$ and obtain a flattened formula. 
 As an example, consider the formula $\varphi=a \until_{[0,3]}(c \wedge \cnt_{(2,3)}^{=1}d \until_{(0,1),\#{(d \wedge \cnt^{=1}_{(0,1)}e) \geq 1}} \cnt^{\geq 2}_{(0,1)}e])$. Replacing the counting modalities with fresh witness propositions $w_1, w_2$, we obtain 
	$\varphi_{flat}=[a \until_{[0,3]} (c \wedge w_1)] \wedge T$  where 
	$T=T_1 \wedge T_2 \wedge T_3\wedge T_4$, with 
	$T_1=\wB[w_1 \leftrightarrow \cnt_{(2,3)}^{=1}w_2]$,  
	$T_2=\wB[w_2 \leftrightarrow d \until_{(0,1), \#{w_4}\geq 1} w_3]]$,
	$T_3=\wB[w_4 \leftrightarrow (d \wedge \cnt^{=1}_{(0,1)}e)]$, and 
	$T_4=\wB[w_3 \leftrightarrow \cnt^{\geq 2}_{(0,1)}e]$. 
Each temporal projection $T_i$ obtained after flattening contains either a 
	$\mathsf{C}$ modality or a $\mathsf{UT}$ modality. 
	In the following, we now show how to 
	obtain equisatisfiable $\mtl$ formulae corresponding to each temporal projection. The proof of Lemma \ref{remove-cntinf} is in Appendix \ref{app:inf}.
	\begin{lemma}
\label{remove-cntinf}
The formula $\cnt^{\geq n}_{\langle l, \infty)}b$ 
has an equivalent formula in $\mtl$.
	\end{lemma}
We now outline the steps followed to obtain an equisatisfiable formula in $\mtl$, assuming $\cnt^{\geq n}_{\langle l, \infty)}b$ modalities 
have been eliminated using Lemma \ref{remove-cntinf}.
 \begin{enumerate}
 \item \textbf{\it {Flattening}} :  Flatten  $\chi$ obtaining $\chi_{flat}$ over $\Sigma \cup W$, where $W$ is the set of witness propositions 
used, $\Sigma \cap W=\emptyset$.
  \item \textbf{\it{Eliminate Counting}} :  Consider, one  by one,
each temporal definition $T_i$ of $\chi_{flat}$.
  	 Let $\Sigma_i=\Sigma\cup W \cup X_i$, where $X_i$ is 
  a set of fresh propositions,  $X_i \cap X_j=\emptyset$ for $i \neq j$. 
 \begin{itemize}
 \item  If $T_i$ is a temporal projection containing a $\mathsf{C}$ modality
of the form $\cnt^{\sim n}_{\langle l, u \rangle}$, or a $\mathsf{UT}$ modality 
of the form $x \until_{I,\#b \le n}y$, then 
Lemma \ref{cnt-sp} 
   synthesizes a   formula $\zeta_i \in \mathsf{MTL}$ over $\Sigma_i$ such that 
   $T_i \equiv \exists X_i. \zeta_i$. 
\item If $T_i$ is a temporal projection containing a $\mathsf{UT}$ modality
of the form $x \until_{I,\#b \ge n}y$, Lemma \ref{th-os}
gives  $\zeta_i \in \mtl$ over $\Sigma_i$ such that  
$ONF_{\Sigma}(T_i)\equiv \exists \downarrow X_i. \zeta_i$. 
 \end{itemize}
 \item \textbf{\it{Putting it all together}} :  The formula
     $\zeta = \bigwedge_{i=1}^k \zeta_i  \in \mtl$ is such that \\ $\bigwedge_{i=1}^kONF_{\Sigma}(T_i) \equiv \exists \downarrow X. \bigwedge_{i=1}^k \zeta_i$ where
    $X=\bigcup_{i=1}^k X_i$. 
           \end{enumerate}
 
\begin{lemma}
 \label{cnt-sp}
 \begin{enumerate}
 \item  Consider a temporal definition 
 $T=\wB[a \leftrightarrow \cnt^{\geq n}_{[l, u)}b]$,  built from $\Sigma \cup W$. 
 Then we synthesize a formula $\zeta \in \mathsf{MTL}$ 
 over  $\Sigma \cup W \cup X$ 
  such that $T  \equiv \exists  X. \zeta$.
\item 	Consider a temporal definition 
	$T=\wB[a \leftrightarrow x \until_{I,\#b \le n}y]$,  built from $\Sigma \cup W$. 
	Then we  synthesize a formula $\zeta \in \mtl$ 
	over  $\Sigma \cup W \cup X$ 
	such that $T  \equiv \exists X.\zeta$.
 \end{enumerate}
 \end{lemma}
\begin{proof}
\begin{enumerate}
\item Lets consider intervals of the form $[l,u)$. Our proof extends to all intervals $\langle l, u \rangle$.
Consider $T=\wB[a \leftrightarrow \cnt^{\geq n}_{[l, u)}b]$. 
Let $\oplus$ denote  addition modulo $n+1$.
\begin{enumerate}
\item \emph{Construction of a ($\Sigma \cup W, X)$- simple extension}. We introduce a fresh set of propositions $X = \{b_0,b_1,\ldots,b_n\}$ and construct a simple extension 
$\rho'=(\sigma', \tau')$ from $\rho=(\sigma,\tau)$ as follows:
\begin{itemize}
	\item \textbf {$C1$}: $\sigma'_1= \sigma_1 \cup \{b_0\}$. If $b_k \in \sigma'_i$   
	 and if $b \in \sigma_{i+1}$, $\sigma'_{i+1}=\sigma_{i+1} \cup \{b_{k \oplus 1}\}$. 
	\item \textbf {$C2$}:  If  $b_k \in \sigma'_i$  and $b \notin \sigma_{i+1}$, then $\sigma'_{i+1}=\sigma_{i+1} \cup \{b_k\}$.
	\item \textbf{$C3$}: $\sigma'_{i}$ has exactly one symbol from $X$ for all $1 \leq i \leq |dom(\rho)|$. 
	\end{itemize}
\item \emph{Formula specifying the above behaviour}. 
 The variables in $X$ help in counting the number of $b$'s in $\rho$. 
 $C1$ and $C2$ are written in $\mathsf{MTL}$ as follows: 
\begin{itemize}
	\item  $\delta_1 = \bigwedge \limits_{k=0}^n \wB[(\nex b\wedge b_k) \rightarrow \nex b_{k\oplus1}]$ and 
 $\delta_2 = \bigwedge \limits_{k=0}^n \wB[(\nex \neg b \wedge b_k) \rightarrow \nex b_k]$
\end{itemize}
\item  \emph{Marking the witness `$a$' correctly at points satisfying $\cnt^{\geq n}_{[l, u)}b$}.  
The index $i$ of $b_i$ at a chosen point gives the number of $b$'s seen so far
since the previous occurrence of $b_0$.
From a point $i$, if the interval $[t_i+l, t_i+u)$ 
has $k$ elements of $X$, then
there must be $k$ $b$'s in   $[t_i+l, t_i+u)$.
 To mark the witness $a$ appropriately, we need to check the number of times $b$ occurs in $[t_i+l,t_i+u]$ from the current point $i$.
 A point $i \in dom(\rho')$ is marked with witness $a$ iff all variables of $X$ are present in $[t_i+l,t_i+u)$, as 
 explained in $\mtl$ by   
  $\kappa = \wB[a \leftrightarrow (\bigwedge \limits_{k = 1}^n \fut_{[l,u)} b_k)]$.
\end{enumerate}
$\zeta = \delta_1 \wedge \delta_2 \wedge \kappa$ in $\mathsf{MTL}$ is equisatisfiable to $T$ modulo simple projections.
\item The proof is similar to the above, 
 details are in Appendix \ref{app:tmtlsp}.
\end{enumerate}

	\end{proof}

\begin{lemma}
	\label{th-os}
	Consider a temporal definition 
	$T=\wB[a \leftrightarrow x \until_{I,\#b \ge n}y]$,  built from $\Sigma \cup W$. 
	Then we synthesize a formula $\psi \in \mtl$ 
	over  $\Sigma \cup W \cup X$ 
	such that $ONF_{\Sigma}(T)\equiv \exists \downarrow X. \psi$ where $ONF_{\Sigma}(T)$ is $T$ relativized with respect to $\Sigma$. 
\end{lemma}
\begin{proof}
	If $I$ is of the form $\langle l,\infty )$, 
	then $x \until_{\langle l, \infty) ,\#b \ge n}y \equiv 
	x \until_{\langle l, \infty)}y \wedge x \until_{\#b \ge n}y$. 
	The untimed threshold formula $x \until_{\#b \ge n}y$ can be rewritten in $\mathsf{LTL}$ \cite{cltl}.
		
		The next case is when the interval $I$ is bounded of the form  $[l,u)$. Our reduction below can be adapted to other 
	kinds of bounded intervals.  
	Let $j$ be any point. 
Let $far_j$ be the farthest point in the $[l,u)$ future of $j$ such that $y$ is true at $far_j$, and  $x$ continuously holds 
at all the intermediate points between $j$ and $far_j$. To check the truth of $x \until_{I,\#b \ge n}y$ at $j$, we need to assert that the number of $b$'s from $j$ to $far_j$ is $\ge n$. We first count the number of $b$'s from the first integer point in the $[l,u)$ future of $j$ (let this point be $\alpha$) to $far_j$ and 
add this to the number of $b$'s between $j$ and $\alpha$. In case $far_j$ lies before $\alpha$, then we simply count the number of $b$'s 
between $j$ and $far_j$. Since we may not have all integer points at our disposal, we 
oversample the model by adding extra points at all integer time stamps.

			Let $L = u-l$. Define $s \boxplus t = min(s+t,n)$, and $s \oplus t=(s+t)~\mbox{mod}~(u+1)$.
		\\1) {\it{Construction of a $(\Sigma \cup W, X)$-oversampled behaviour}}. We introduce a fresh set of propositions $X = C \cup A \cup B$ where $C,B,A$ are  defined below. Given any timed word $\rho$, we then construct a  $(\Sigma \cup W, X)$-oversampled behaviour $\rho'=(\sigma', \tau')$ from $\rho=(\sigma, \tau)$.	\begin{itemize}
		\item \textbf{$O1$:} $C = \{c_0,c_1,\ldots,c_u\}$. A point $i$ of $\rho$ is marked $c_g$ iff $t_i~\mbox{mod}~u = g$. 
		In the absence of such a point $i$  
		 (such that $t_i$ is an integer value $k < t_{|dom(\rho)|}$), we add a new point $i$ to $dom(\rho)$ with time stamp $t'_i$ and mark it with $c_g$ iff $t'_i~\mbox{mod}~u = g$. 
		 Let $\rho_c=(\sigma^c, \tau^c)$ denote the word obtained from $\rho$ after this marking. 
		\item \textbf{$O2$:} $B = \cup_{i=0}^{u}B^i$, where $B^i=\{b^i_0,b^i_1,\dots b^i_n\}$. 
		 All the points of $\rho_c$ marked $c_i$ are marked as $b^i_0$. Let $p,q$ be two integer points such that $p$ is marked  
		   $c_i$, $q$ is marked $c_{i \oplus L}$, and no point between $p, q$ is marked  $c_{i \oplus L}$. 
		   $p, q$ are $L$ apart from each other. 		   		  Let $p < r <  q$ be  such that  $b^i_g \in \sigma^c_r$ for some $g$. If 
		  $c_{i \oplus L} \notin \sigma^c_{r+1}$ and $b \in  \sigma^c_{r+1}$, then 
		  the point $r+1$ is marked $b^i_{g \boxplus 1}$. 
		  If $c_{i \oplus L}, b \notin \sigma^c_{r+1}$, then 
		  the point $r+1$ is marked $b^i_g$.
		  Each $B^i$ is a set of counters which are reset at $c_i$ and counts the number of occurrences of $b$ upto the threshold $n$
		between a $c_i$  and the next occurrence of $c_{i \oplus L}$. 
		Starting at a point marked $c_i$ with counter $b^i_0$, the counter increments upto $n$ on encountering a $b$, until the next $c_{i \oplus L}$. 
		Further, we ensure that the counter $B^i$ does not appear anywhere from $c_{i\oplus L}$ to the next $c_i$.  
		 Let the resultant word be $\rho_b$.
	\item \textbf{$O3$:} $A = \{a_0,a_1,\ldots,a_n\}$.  Consider any point $j$ in $\rho_b$ with time stamp $t_j$. 
	Let $\alpha$ be a point with time stamp $\lceil t_j+l \rceil$. 
	Let $max_j$ represent a point satisfying the following conditions: (a) 
 $y$ is true at $max_j$ and  $t_{max_j} \in [t_j+l, t_j+u)$, (b) $x$ is true at all points between $j$ and $max_j$, and (c) the number of occurrences of $b$ 
 from $\alpha$ to $max_j$ is either $\geq n$, or is the maximum amongst all points which satisfy 
 (a) and (b).   The point $j$ is marked $a_h$ iff $h<n$ is the number 
 of occurrences of $b$'s from $\alpha$ to $max_j$. If the count of $b$'s from $\alpha$ to $max_j$ is $\geq n$, then 
 $j$ is marked $a_n$.   Note that whenever $max_j$ exists, it will be at or after $\alpha$. 
 $max_j$ need not always exist; we could have a point 
      $\beta$ with time stamp $t_j \leq t_{\beta} \leq t_{\alpha}$ such that   
         $y$ is true at $\beta$, $x$ holds continuously between $j$ and $\beta$, and the number of occurrences of $b$ in between $j$ and $\beta$ is $\ge n$.
  Let $\rho'$ be the word obtained after all the markings.
	\end{itemize}
2) {\it{Formula for specifying above behaviour}}. We give following $\mathsf{MTL}$ formulae to specify $O2$ and $O3$. 
 $\delta_2=\bigwedge \limits_{i=0}^u (\delta_{2i}(1) \wedge \delta_{2i}(2) \wedge \delta_{2i}(3))$ encodes $O2$ where \\
$ \delta_{2i}(1) = \wB(c_i \rightarrow b^i_0) \wedge \bigwedge \limits_{k=0}^n \wB[(\nex (b\wedge \neg c_{i \oplus L})\wedge b^i_k) \rightarrow \nex b^i_{k\boxplus1}]$,\\
 $\delta_{2i}(2) = \bigwedge \limits_{k=0}^n \wB[(\nex (\neg b \wedge \neg c_{i \oplus L}) \wedge b^i_k) \rightarrow \nex b^i_{k}]$ and \\
$\delta_{2i}(3)=\bigwedge \limits_{i=0}^u \wB[c_{i \oplus L} \rightarrow (\neg c_i \wedge \neg b^i) \until c_i]$, where 
$b^i=\bigvee \limits_{k=0}^u b^i_k$.   
$O3$ is encoded by $\delta_3 = \bigvee\limits_{i=0}^{i=u}(\wB[a_h \leftrightarrow (\neg act \vee x) \until_{[l,u)} (y \wedge b^i_h) \wedge  \neg \{(\neg act \vee x) \until_{[l,u)} (y \wedge b^i_{h+1})\}] \wedge \fut_{\mathcal{I}}c_i)$ where $\mathcal{I}=[l,l+1)$.   The truth of $\delta_3$ relies on the fact that if $x \until_{[l,u)} (y \wedge b^1_h)$ and $x \until_{[l,u)} (y \wedge b^i_{h+2})$ are both  true at a point, then $x \until_{[l,u)} (y \wedge b^i_{h+1})$ is also true at the same point. 	 Hence,  if $x \until_{[l,u)} (y \wedge b^i_h)$ is true,  and $x \until_{[l,u)} (y \wedge b^i_{h+1})$ is not true at some point, then $h$ is the largest number such that $x \until_{[l,u)} (y \wedge b^i_h)$ is true. Let $act=\bigvee(\Sigma \cup W)$. 
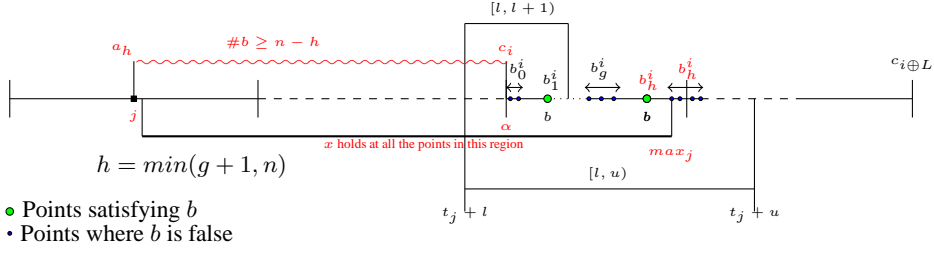
\begin{figure}[t]
\begin{tikzpicture}
\foreach \sc in {1.1}
\foreach \x in {0.25}{
\draw (\x,0) -- (\x+3*\sc,0);
\draw[dashed] (\x+3*\sc,0)--(\x+6*\sc,0);
\draw (\x+6*\sc,0)--(\x+6.5*\sc,0);
\draw[dotted] (\x+6.5*\sc,0)--(\x +7*\sc,0);
\draw (\x+7*\sc,0)--(\x+8.35*\sc,0);
\draw[dashed] (\x+8.35*\sc,0) -- (\x+9.5*\sc,0);
\draw (\x + 9.5*\sc,0) -- (\x + 12,0);
\foreach \y in {0*\sc,3*\sc,6*\sc,9,12}
{
\draw (\x+\y,-0.25)--(\x+\y,0.25);
}


%
%
%
%
%

%
\draw (\x+5.5*\sc,1)--(\x+5.5*\sc,-1.5);
\node at (\x + 6.2*\sc,1.2) {\tiny{$[l,l+1)$}};
\draw (\x+6.75*\sc,1)--(\x+6.75*\sc,0);
\draw (\x+5.5*\sc,1)--(\x+6.75*\sc,1);
\draw (\x+9*\sc,0)--(\x+9*\sc,-1.5);

\draw (\x + 5.5*\sc, -1.2) -- (\x + 9*\sc, -1.2);
\draw (\x + 1.6*\sc, 0) -- (\x + 1.6*\sc, -0.5);
\draw (\x + 8*\sc, 0) -- (\x + 8*\sc, -0.5);
\node at (\x + 8*\sc,-0.75) {\tiny {\textcolor{red}{$max_j$}}};
\draw [thick] (\x + 1.6*\sc, -0.5) -- (\x + 8*\sc, -0.5);
\node at (\x + 6*\sc,-0.38) {\tiny {\textcolor{red}{$\alpha$}}};
\node at (\x + 6*\sc,0.65) {\tiny {\textcolor{red}{$c_i$}}};
\node at (\x + 12,0.45) {\tiny{$c_{i  \oplus L}$}};
\node at (\x+ 5*\sc,-0.61) {\tiny{\textcolor{red}{$x$ holds at all the points in this region}}};
\node at (\x + 7.23*\sc,-1) {\tiny {$[l,u)$}};
\node at (\x+ 5.5*\sc,-1.55) {\tiny{$t_j+l$}};

\node at (\x+ 9*\sc,-1.55) {\tiny{$t_j+u$}};


\foreach \y in {1.5*\sc,6*\sc}
{
	\draw (\x+\y,0)-- (\x+\y,0.5);
}
\node at (\x+3.81*\sc,0.5) {\textcolor{red}{\uwave{\hspace{5cm}}}};
\foreach \y in {1.5*\sc}
{
\node[fill = black,draw = black,rectangle,inner sep=1pt,label=below:{\tiny{\textcolor{red}{$j$}}}] at (\x+\y,0){};
}

\foreach \w in {0*\sc}
{

\foreach \y in {7*\sc,7.15*\sc,7.3*\sc, 8*\sc,8.1*\sc,8.25*\sc,8.35*\sc}
{
	\node[fill =blue,draw = black,circle,inner sep=0.5pt] at (\x+\y+\w ,0){};
	
}

\draw[<->] (\x+6.95*\sc+\w*\sc,0.15)--(\x+7.38*\sc+\w*\sc,0.15);
\node at (\x + 7.13*\sc,0.4) {\tiny{$b^i_g$}}; 
\draw[<->] (\x+7.95*\sc+\w*\sc,0.15)--(\x+8.38*\sc+\w*\sc,0.15);
\node at (\x + 8.2*\sc,0.4) {\tiny{\textcolor{red}{$b^i_h$}}};
\node at (7.95*\sc,0.25) {\tiny{\textcolor{red}{$b^i_h$}}};



}
\foreach \y in {7.7}
{
	\node[fill = green,draw = black,circle,inner sep=1pt,label=below:{\tiny{$b$}}] at (\x+\y*\sc,0){};

	\node[fill = green,draw = black,circle,inner sep=1pt,label=below:{\tiny{$b$}}] at (\x+\y*\sc,0){};
}

\foreach \y in {6.05*\sc,6.15*\sc}
{
	\node[fill = blue,draw = black,circle,inner sep=0.5pt] at (\x+\y ,0){};
	
}
\draw[<->] (\x+6*\sc,0.15)--(\x+6.2*\sc,0.15);
\node at (\x + 6.15*\sc,0.4) {\tiny{$b^i_{\tiny{0}}$}};
\node at (\x + 6.55*\sc,0.25) {\tiny{$b^i_{\tiny{1}}$}};
	\node[fill = green,draw = black,circle,inner sep=1pt,label=below:{\tiny{$b$}}] at (\x+6.5*\sc,0){};

\node at (\x+2.2*\sc,-.9) {$h = min(g+1,n)$}; 

\node at (3.75,0.75){\tiny{\textcolor{red}{$\#b \ge n-h$}}};
	\node[fill = green,draw = black,circle,inner sep=1pt,label=right:{Points satisfying $b$}] at (\x+0,-1.5){};
	\node[fill = blue,draw = black,circle,inner sep=0.5pt,label=right:{Points where $b$ is false}] at (\x+0,-1.791){};
\node at (1.75,0.65) {\tiny {\textcolor{red}{$a_h$}}};

}

\end{tikzpicture}
\caption{Illustration of point $j$, $max_j$ and the point $\alpha$ such that $t_{\alpha}=\lceil t_j+l \rceil$.
$\alpha$ is marked with some $c_i$ since it is an integer time point. 
The counting of $b$'s is reset at $c_i$, starting with $b^i_0$, and continues till $c_{i \oplus L}$.
Since $max_j$ is marked $b_h^i$, $j$ is marked $a_h$. $h$ is the count of $b$'s between $\alpha$ and $max_j$.
To satisfy $\Box(a \leftrightarrow x \mathsf{U}_{[l,u),\#b \ge n} y)$ at $j$, we check that the number of $b$'s 
between $j$ and $\alpha$ is $\geq n-h$ when $b$ is not true at $\alpha$, and is $\geq n-h-1$ when $b$ is true at $\alpha$.}  
\label{tmtl-fig}
\end{figure}

\noindent 3) \emph{Marking the witness `a' correctly at points satisfying} $x \until_{I,\#b \ge n}y$. 
Let $j$ be any point in $\rho'$, such that $max_j$ exists. 
 We first count the number of $b$'s from  $j$ to the farthest integer point $\alpha$ (recall that $t_{\alpha}= \lceil t_j + l\rceil$), followed by counting 
 the number of $b$'s from $\alpha$ to $max_j$. 
  Note that the index $h$ of $a_h$ marked at $j$ gives the count (upto $n$) of $b$'s from $\alpha$ to $max_j$. 
  We check the count of $b$'s between $j$ and $\alpha$ is $\geq n-h$. Let $\mathcal{I}=[l,l+1)$.\\
\noindent	$\lambda_1 = \bigvee \limits_{h=0,i=0}^{h=n,i=u} [(a_h \wedge \fut_{\mathcal{I}}(c_i \wedge \neg b) \wedge [(\neg act \vee x) \wedge \neg c_i]
	\until_{\#b\ge n-h} c_i)]$\\
\noindent	$\lambda_2 = \bigvee \limits_{h=0,i=0}^{h=n,i=u} [(a_h \wedge \fut_{\mathcal{I}}(c_i \wedge b) \wedge [(\neg act \vee x) \wedge \neg c_i]
	\until_{\#b\ge n-h-1} c_i)]$\\
If $max_j$ does not exist, then we characterize the point  $\beta$  by the truth of the formula $\lambda_3=((x \vee \neg act) \wedge \neg c) \until_{\#_b \ge n} y$, where $c=\bigvee c_k$. The formula $\lambda=\wB[a \leftrightarrow (\lambda_1 \vee \lambda_2 \vee \lambda_3)]$
	captures marking point $j$ correctly with $a$.
    Thus we obtain the $\mathsf{MTL}$ formula
    $\zeta = \delta_2 \wedge \delta_3 \wedge \lambda$.\\
    \end{proof} 
\section{Discussion and Related Work}
Within temporal and real time logics, the notion of counting has attracted considerable interest.
Laroussini {\em et al} extended untimed $\mathsf{LTL}$ with threshold counting constrained until operator. They showed that the expressiveness of 
$\mathsf{LTL}$ is not increased by adding
threshold counting but the logic become exponentially more succinct. 
Hirshfeld and Rabinovich introduced $\cnt_{(0,1)}$ operator in continuous timed $\mathsf{QTL}$ and showed that it added expressive power. They also showed that in 
continuous time, more general $\cnt_{\langle l,u \rangle}$ operator can be expressed with just $\cnt_{(0,1)}$. Building upon this, Hunter showed that $\mathsf{MTL}$ with $\cnt_{(0,1)}$ operator
is expressively complete w.r.t. $\mathsf{FO}[+,1]$. Thus it can also express $\mathsf{UT}$ operator which is straightforwardly modelled in $\mathsf{FO}[+,1]$.

In this paper, we have explored the case of $\mathsf{MTL}$ with counting operators over timed words interpreted in pointwise manner. 
We have shown that both $\cnt_I$ and $\mathsf{UT}$  operators add expressive power to $\mathsf{MTL}$. Moreover, the two operators are independent in the sense 
that neither can be expressed in terms of the other and $\mathsf{MTL}$. (We use prefixes $\cnt$ and $\mathsf{T}$ to denote a logic extended with $\mathsf{C}$ and $\mathsf{UT}$ operators respectively).
It is easy to show (see Appendix \ref{app:tptl}) that $\mathsf{CTMTL} \subset \mathsf{TPTL}^1$.
All these expressiveness results straightforwardly carry over to $\mtl$ over infinite timed words.
Thus, pointwise semantics exhibits considerable complexity in expressiveness of operators as compared to continuous time semantics where all these logics
are equally expressive. While this may arguably be considered a shortcoming of the pointwise models of timed behaviours, the pointwise models have superior 
decidability properties making them more amenable to algorithmic analysis. $\mathsf{MTL}$ already has undecidable satisfiability in continuous time whereas it
has decidable satisfiability over finite timed words in pointwise semantics.

In this paper, we have shown that $\mtl$ extended with $\cnt$ and $\mathsf{UT}$ operators also has decidable satisfiability. The result is proved by giving an 
equisatisfiable reduction from $\mathsf{CTMTL}$ to $\mathsf{MTL}$ using the technique of oversampling projections. This technique was introduced earlier \cite{time14} and used to show that $\mathsf{MTL}[\until_I,\since_{\mathsf{np}}]$  with non-punctual past operator is also decidable in pointwise semantics. 
Current paper marks one more use of the technique of oversampling projections. 
A closer examination of our reduction from $\mathsf{CTMTL}$ to $\mathsf{MTL}$ shows that it can be used in presence of any other operator. Also, it does not introduce any 
punctual use of $\until_I$ in reduced formula. The reduced formula is exponentially larger than the original formula 
(assuming binary encoding of integer constants). All this implies that $\mathsf{CTMTL}[\until_I,\since_{\mathsf{np}}]$ is also decidable over finite timed words.
Moreover, $\mathsf{CTMITL}[\until_{\mathsf{NS}},\since_{\mathsf{NS}}]$ can be equisatisfiably reduced to $\mathsf{MITL}[\until_{\mathsf{np}},\since_{\mathsf{np}}]$ and it is decidable with at most 2-$\mathsf{EXPSPACE}$ complexity. 
The exact complexity of satisfiability checking of $\mathsf{CTMITL}$ is open although $\mathsf{EXPSPACE}$ lowerbound trivially follows from $\mathsf{MITL}$ and counting $\mathsf{LTL}$ which 
are syntactic subsets.

In another line of work involving counting and projection,
Raskin \cite{raskin-thesis} extended $\mathsf{MITL}$ and event clock logic with ability to count by extending these logics with automaton operators and adding second order quantification. The expressiveness was shown to be that of  recursive event clock automaton. These logics were able to count over 
the whole model rather than a particular timed interval. The resultant logic cannot specify constraints like within a time unit $(0,1)$ the number of 
occurrence of a particular formula is $k$ but can also incorporate mod counting. Thus Raskin's logics and the $\mathsf{CTMTL}$ are expressively independent.

\newpage
\bibliography{papers}

\newpage
\appendix
\centerline{\Large{Appendix}}
\section{Motivation}

\begin{example}
Our first example is motivated from medical devices used in monitoring foetal heart rate.
 In neo-natal care, the use  
  of external and 
internal foetal heart rate monitoring devices is well-known. 
The average foetal heart rate is between 110 and 160 beats per minute, and can vary 5 to 25 beats per minute. 
 An abnormal foetal heart rate ($<$ 100 beats per minute or $>180$ beats per minute) may indicate that the foetus is not getting enough oxygen or that there are other problems.  Current techniques rely predominantly on the use of electronic foetal monitoring through the use of cardiotocography (CTG). This technique records changes in the foetal heart rate (FHR) (via Doppler ultrasound or direct foetal ECG measurement with a foetal scalp electrode) and their temporal relationship to myometrial activity and uterine contractions. In  high risk cases, the electronid foetal monitoring 
 is combined with checking the mother's blood oxygen saturation levels.
 Normal blood oxygen levels are considered 95-100 percent. 
 These are specialised real-time properties that need to be formally specified 
  in order to model check important safety properties in medical devices.  
 These properties are not only time critical, but also need to measure the 
 number of times an event occurs in a given interval of time, to 
  ensure safety.  
  Let the proposition $\mathsf{fhb}$ denote a foetal heart beat, and 
 let the  proposition $\mathsf{sp-ok}$ denote normal blood oxygen levels of the mother. 
 The $\mathsf{CMTL}$ formula   $\Box_{[0,60]}\mathsf{spo{-}ok} \wedge \cnt^{\geq 110}_{[0,60]}(\mathsf{fhb}) \wedge 
  \cnt^{\leq 160}_{[0,60]}(\mathsf{fhb})$ specifies that in a duration of 60 seconds, the mother's blood oxygen levels are normal, while the foetal heart beats  in the range of [110,160].

    \end{example}
  
  \begin{example}
  Our second example is motivated from the problem of energy peak reduction 
  in large organisations using HVAC systems. 
  The problem of energy peak demand reduction within a large organization by synchronizing switching decisions of various ``heating, ventilation, and air conditioning'' (HVAC) systems is one of the 
  most practically relevant ones. The relationship between energy demand peaks and extreme climatic conditions 
  has been studied in the literature; hence, reducing the energy peak demand of HVAC systems can significantly 
   reduce the power peak demand.  Nghiem et al. considered the model of an organization 
   divided into various zones,  where at any given point of time, the HVAC system of a zone can be switched off or switched on  to ensure that the zone stays in a comfortable temperature range.          
  Several scheduling algorithms  for the same have been proposed so far in the literature, with the restriction that  simultaneously a bounded number of HVAC systems are switched on at any point in time.
Also, the number of times a HVAC unit oscillates between the on and off mode should be 
minimal, while respecting the comfortable temperature range in each zone. 
 Synthesizing the optimal number of HVAC units that have to remain switched on 
 to maintain the comfort level in any zone is an important research problem. 
We motivate the use of our counting logics to  specify the number of times an HVAC unit switches between the on and the off  mode. Let 
$\mathsf{hvac}^i_1$ be a proposition that evaluates to true when a HVAC in zone $i$  has just been switched on, and let $\mathsf{hvac}^i_1$ be a proposition that evaluates to true when a HVAC in zone $i$  has just been switched off. Let $I$ be the set of zones and 
  $\mathsf{zone}^i_{\mathsf{high\_temp}}$  be a proposition which evaluates to true when a zone $i$ is not in its comfortable 
 temperature zone, and let $\mathsf{zone}^i_{\mathsf{cz}}$
 be a proposition which evaluates to true when a zone $i$ is in its comfortable 
 temperature zone. Let $\psi_x= \mathsf{hvac}^i_{1-x} \wedge \neg \mathsf{hvac}^i_{x}~ \until ~\mathsf{hvac}^i_{x}$ and let  $\eta=\bigvee \limits_{x \in \{0,1\}} \psi_x$.   
  The $\mathsf{TMTL}$ 
 formula  $\bigwedge \limits_{i \in I}\wB(\mathsf{zone}^i_{\mathsf{high\_temp}} \rightarrow \fut_{\{
  [0,u],\#_\eta \le n\}} (\mathsf{zone}^i_{\mathsf{cz}}))$  specifies that any zone which is not in the comfortable range should reach the comfort zone in no more than $u$ time units,  and while reaching there the number of switches from on to off or off to on of any HVAC in the zone is at most $n$ times. One may also want to control the average number of times the switching happens between on and off. The $\mathsf{CMTL}$ formula  $\wB[\cnt^{<c}_{[0,1]} \bigvee \limits_{i \in I}\mathsf{hvac}^i]$ where $\mathsf{hvac}^i = \mathsf{hvac}^i_1 \vee \mathsf{hvac}^i_1$ specifies that from any event within $[0,1]$ the number of times any HVAC is switched on or off  is $<c$.
These counting logics can be used to model check the HVAC scheduling algorithms;  
it is also possible to rewrite these algorithms in the counting logics. 
Satisfiability checking of these logics can then be used to find 
the optimal number $k$ of HVAC systems that are required to be on to ensure a comfort temperature range in any zone
for a given time interval.  
Assuming that the environment behaviour and the scheduling algorithm is given in some declarative form, 
satisfiability checking of the formula 
$\fut_{\{[l,u],\#\mbox{switches}<w\}}(\bigwedge \limits_{i\in I} \neg \mathsf{zone}^i_{\mathsf{high\_temp}} \rightarrow \mathsf{Algo}_k \wedge \mathsf{environment\_parameter})$
for various values of $k$ and finding the minimal 
such $k$ tells the optimal number of HVAC units that should remain switched on per zone.
 \end{example}

\section{$\mathsf{CTMTL}$ with General Threshold Formulae}
\label{app:tmtl}
We now generalize the threshold modality used in $\mathsf{CTMTL}$ as follows: For $\varphi \in \mathsf{CTMTL}$, and $\sim \in \{<, \leq, \geq, >\}$
$$\eta:=\#\varphi \sim c~|~\eta \wedge \eta~|~\eta \vee \eta~|\neg\eta$$
We show in this section that any formula in $\mathsf{CTMTL}$ that is written with a 
complex threshold formula can be rewritten in terms of simple threshold formulae as introduced in Section \ref{sec:ctmtl}
while introducing $\mathsf{CTMTL}$.
\begin{lemma}
\label{cnf-open}
Let $\varphi \until_{I, \eta} \psi \in \mathsf{CTMTL}$ with 
 $\eta=\eta_1 \vee \dots \vee \eta_n$. Then 
 $\varphi \until_{I, \eta} \psi $ is equivalent to 
  $\varphi \until_{I, \eta_1} \psi \vee  \varphi \until_{I, \eta_2} \psi \vee \dots \vee  \varphi \until_{I, \eta_n} \psi$.  
\end{lemma}
 \begin{proof}
 Let $\eta_i=\#{\varphi_i}\sim m_i$ for $1 \leq i \leq n$.
 Given a timed word $\rho$ and a point $i \in dom(\rho)$,
 $\rho, i \models \varphi \until_{I, \eta} \psi$ iff there is a point $j >i$ such that
 $\rho, j \models \psi$,  and $\varphi$ evaluates to true at all the in between points 
 $i < k < j$, and there is atleast one formula $\varphi_i$ such that the number of points between $i$ and $j$ 
 where $\varphi_i$ evaluates to true is $\sim m_i$. 
 Hence we obtain $\rho, i \models \varphi \until_{I, \eta_r} \psi$, for some $1 \leq r \leq n$.
  The converse is similar.
 \end{proof}                                                         
    A threshold formula $\eta$ 
 is called atomic iff all the threshold formulae
 $\eta_j$ occurring in $\eta$ 
 cannot be written as
the conjunction or disjunction of two threshold formulae. 
Thus, the threshold formula  $\eta=\#(a \until_{J, \#b=5 \wedge \#c < 3}) \geq 5$  is not atomic, since 
it involves a conjunction of two threshold formulae. 

It can be easily seen that every threshold formula $\eta$ is equivalent to some threshold formula 
 $\eta_1$ in disjunctive normal form. 
 The formula $\eta_1$ in DNF 
    is  obtained by recursively replacing 
 all the threshold formulae occurring in  $\eta$ in DNF.

   For instance, $\eta=[\#b=5] \wedge [\#(a \until_I b) < 7] \vee
[\#(a \until_{\{J, \#d<12 \wedge \#e<6\}}c) \geq 2]$ can be expressed  as
  [$\#b=5 \wedge \#(a \until_I b) < 7] \vee 
 [\#[(a \until_{\{J, \#d <12\}} c)  \wedge (a \until_{\{J, \#e < 6\}}c)] \geq 2].$\footnote{see the proof in the next paragraph,  Lemma \ref{final}} 
    Without loss of generality, we assume henceforth that 
 every $\mathsf{UT}$ modality   $\varphi \until_{I, \eta} \psi$ we encounter in $\mathsf{CTMTL}$ formulae
 has $\eta$ in DNF.  The following two lemmas on monotonicity 
 of counting with respect to time are easy to follow.

\begin{lemma}
\label{mon1}
Let $\delta = \#{\delta_1} <  n_1 \wedge \dots \#{\delta_m} < n_m$ be a 
threshold formula. Let $\rho$ be a timed word  and let $x < y$ be two points
in $dom(\rho)$.  Assume that  
$|\rho[x,y](\delta_i)| < n_i$ for all $1 \leq i \leq m$. 
Then  for any $y' \in dom(\rho)$, with $x<y'<y$, we have 
$|\rho[x,y'](\delta_i)| < n_i$ for all $1 \leq i \leq m$.
 \end{lemma}

\begin{lemma}
\label{mon2}
Let $\delta = \#{\delta_1} \geq  n_1 \wedge \dots \#{\delta_m} \geq  n_m$ be a 
threshold formula. Let $\rho$ be a timed word  and let $x < y$ be two points
in $dom(\rho)$. 
Assume that  
$|\rho[x,y](\delta_i)| \geq  n_i$ for all $1 \leq i \leq m$. 
Then  for any $y' \in dom(\rho)$, with $y'>y$, we have 
$|\rho[x,y'](\delta_i)| \geq n_i$ for all $1 \leq i \leq m$.
 \end{lemma}

                     
\begin{lemma}
\label{final}
Let $\varphi \until_{I, \eta} \psi \in \mathsf{CTMTL}$.  
Let $\eta=\alpha \wedge \beta$,
where $\alpha =\#{\alpha_1} \geq  n_1 \wedge \dots \#{\alpha_m} \geq n_m$ and 
$\beta= \#{\beta_1} <  k_1 \wedge \dots \#{\beta_p} < k_p$.
Then  $\varphi \until_{I, \eta} \psi$ is equivalent to 
$\varphi \until_{[0, \infty), \eta} \psi \wedge \bigwedge_{i=1}^m\varphi \until_{I, \#{\alpha_i} \geq n_i} \psi
\wedge \bigwedge_{i=1}^p\varphi \until_{I, \#{\beta_i} < k_i} \psi$.
\end{lemma}
\begin{proof}
Let $\eta=\alpha \wedge \beta$,
where $\alpha$ is the conjunction 
of all the threshold formulae with comparison operator $\geq$ occurring 
in $\eta$ and $\beta$ 
is the conjunction 
of all the threshold formulae with comparison operator $<$ occurring 
in $\eta$.
Let $\alpha=\#_{\alpha_1} \geq  n_1 \wedge \dots \#_{\alpha_m} \geq n_m$ and 
$\beta= \#_{\beta_1} <  k_1 \wedge \dots \#_{\beta_p} < k_p$.
Let $\rho$ be a timed word and let $i \in dom(\rho)$.
$\rho, i \models \varphi \until_{I, \eta} \psi$ iff there is a point $j >i$ with 
$t_j \in t_i+I$, and the number of points in between $i$ and $j$ 
where $\alpha_r$ evaluates to true is $\geq n_r$ 
for $1 \leq r \leq m$ and
 the number 
of points between $i$ and $j$ where $\beta_l$ evaluates to true is $<k_l$ for $1 \leq l \leq p$.

It is easy to see that $\varphi \until_{I, \eta} \psi \rightarrow 
\varphi \until_{[0, \infty),\eta} \psi \wedge \varphi \until_{I, \alpha} \psi
\wedge \varphi \until_{I, \beta} \psi$. 
 
   Conversely, assume that
$\rho, i \models  \varphi \until_{\eta} \psi \wedge \varphi \until_{I, \alpha} \psi
\wedge \varphi \until_{I, \beta} \psi$.
\begin{enumerate}
\item Since $\rho, i \models \varphi \until_{\eta} \psi$,  
there is a point $j_1 > i$ such that 
$\psi$ evaluates to true at $j_1$, $\varphi$ evaluates to true at all points between 
$i$ and $j_1$,
and $|\rho[i,j_1](\alpha_r)| \geq n_r$ for all $1 \leq r \leq m$ and 
$|\rho[i,j_1](\beta_l)| < k_l$ for all $1 \leq l \leq p$. 
\item Since $\rho, i \models  \varphi \until_{I, \alpha} \psi$, 
there is a point $j_2 > i$ such that 
$\psi$ evaluates to true at $j_2$, $\varphi$ evaluates to true at all points between 
$i$ and $j_2$, 
$t_{j_2}-t_i \in I$, 
and $|\rho[i,j_2](\alpha_r)| \geq n_r$ for all $1 \leq r \leq m$.
\item  Since $\rho, i \models  \varphi \until_{I, \beta} \psi$, there is a point $j_3>i$ such that
 $\psi$ evaluates to true at $j_3$, $\varphi$ evaluates to true at all points between 
$i$ and $j_3$, 
$t_{j_3}-t_i \in I$, 
and 
$|\rho[i,j_3](\beta_l)| < k_l$ for all $1 \leq l \leq p$.
\end{enumerate}
Assume $j_2 \leq j_3$. 
We will check whether $j_2 > i$ is the point 
which satisfies all the conditions required 
with respect to $I, \alpha$ and $\beta$. 
Since the number of points between $i$ and $j_3$ 
where $\beta_l$ evaluates to true is $<k_l$ 
for all $1 \leq l \leq p$, and $j_2 \leq j_3$, 
by monotonicity of time (Lemma \ref{mon1}),   
the
number of points between $i$ and $j_2$ where $\beta_l$ evaluates to true is $<k_l$
for all $1 \leq l \leq p$. Also, we know that  
the number of points between 
$i$ and $j_2$ 
where $\alpha_r$ evaluates to true is $\geq n_r$ for all 
$1 \leq r \leq m$. 
  Then indeed we have $\rho, i \models 
\varphi \until_{I, \eta} \psi$. 


Consider the case when  $j_3 < j_2$.  
Assume that there is some $\beta_l$ such that 
$|\rho[i,j_2](\beta_l)| \geq k_l$.  Since we know that $j_1 > i$ is a point
such that $|\rho[i,j_1](\beta_l)| < k_l$ for all $1 \leq l \leq p$, 
it must be that $j_1 < j_2$.
If there is some $\alpha_h$ such that 
$|\rho[i,j_3](\alpha_h)| < n_h$, then again by monotonicity of time (Lemma \ref{mon2}), we know that $j_1 \geq j_3$. So we have 
$j_3 \leq j_1 < j_2$. Hence, $t_{j_1} \in t_i+I$ since 
$t_{j_3}-t_i \in I$ and
$t_{j_2}-t_i \in I$. Thus, we have a point $j_1 > i$ such that 
$t_{j_1}-t_i \in I$, satisfying all the conditions. 
Hence, $\rho,i \models \varphi \until_{I, \eta} \psi$.

Now we show that $\rho,i \models \varphi \until_{I, \alpha} \psi$ iff 
$\rho, i \models \bigwedge_{i=1}^m\varphi \until_{I, \#_{\alpha_i} \geq n_i} \psi$.
The equivalence of $\varphi \until_{I, \beta} \psi$ and 
$\bigwedge_{i=1}^p\varphi \until_{I, \#_{\beta_i} < k_i} \psi$
is similar.

Assume $\rho, i \models \varphi \until_{I, \alpha} \psi$. 
Then it is easy to see that $\rho, i \models \bigwedge_{i=1}^m\varphi \until_{I, \#_{\alpha_i} \geq n_i} \psi$. Conversely, assume that
$\rho, i \models \bigwedge_{i=1}^m\varphi \until_{I, \#_{\alpha_i} \geq n_i} \psi$. 
Then there are points $j_1, \dots, j_m > i$ such that 
$t_{j_i}-t_i \in I$, $\psi$ evaluates to true at $j_i$, 
$\varphi$ evaluates to true at all points between $i$ and $j_i$, and
$|\rho[i,j_i](\alpha_i)| \geq n_i$ for all $1 \leq i \leq m$.
Let $j_k$ be the point among $j_1, \dots, j_m$ that is farthest from $i$. Then 
clearly,  by monotonicity of time (Lemma \ref{mon2}), 
$|\rho[i,j_k](\alpha_i)| \geq n_i$ for all $1 \leq i \leq m$. Hence, 
$j_k >i$ is the point which satisfies all the conditions 
required of  $\varphi \until_{I, \alpha} \psi$, and hence,
$\rho,i \models \varphi \until_{I, \alpha} \psi$.
\end{proof}                    

\section{Recalling $\mathsf{MTL}$ games from \cite{concur11}}
\label{app:mtlgame}

An $r$-round $I_{\nu}$ $\mathsf{MTL}$ game is played between two players ($\spl$ and  $\dpl$) on  a pair of timed words $(\rho_1,\rho_2)$, where $I_{\nu}$ is the set of intervals allowed in the game.  A configuration of the game is a pair of points $i_p, j_p$ where $i_p \in dom(\rho_1)$ and $j_p \in dom(\rho_2)$. A configuration is called partially isomorphic, denoted $isop(i_p, j_p)$ iff $\sigma_{i_p} = \sigma_{j_p}$.
The starting configuration is $(i_1, j_1)$.    Either $\spl$ or $\dpl$ eventually wins the game. 
 A 0-round game is won by the $\dpl$ iff $isop(i_1,j_1)$. 
	 The $r$ round game is played by first playing one round from the starting position. Either the $\spl$ wins the round, and the game is terminated or the $\dpl$ wins the round, and now the second round is played from this new configuration and so on. The $\dpl$ wins the game only if he wins all the rounds. 
	 The following are the rules of the game in any round. Assume that the configuration  at the start of the $p$th round is $(i_p, j_p)$.
	\begin{itemize}
		\item If $isop(i_p,j_p)$ is not true, then $\spl$ wins the game, and the game is terminated.  Otherwise, the game continues as follows: 
		\item The $\spl$ chooses one of the words by choosing $\rho_x,  x \in \{1,2\}$. 
		$\dpl$ has to play on the other word $\rho_y$, $x \neq y$.  Then $\spl$ chooses  the  $\until_I$  move, along with the interval $I \in I_{\nu}$ (such that the end points of the intervals are non-negative integers). Given the current configuration as $(i_p,j_p)$, the rest of the $\until_I$ round is played as follows:
		\begin{itemize}
			\item $\spl$ chooses a position $i'_p \in dom(\rho_x)$ such that $i_p < i'_p$  and $(t_{i'_p} - t_{i_p}) \in I$.
			\item The $\dpl$ responds to the $\until_I$ move by choosing $j'_p\in dom(\rho_y)$ in the other word such that $j_p < j'_p$ 
			and $(t_{j'_p} - t_{j_p}) \in I$.   
			If the $\dpl$ cannot find such a position, the
			$\spl$ wins the round and the game. Otherwise, the game continues and $\spl$ chooses one of the following options.
			\item $\fut$ Part: The round ends with the configuration $(i'_p, j'_p)$.
			\item $\until$ Part: $\spl$ chooses a position $j''_p$ in $\rho_y$ such that $j_p < j''_p < j'_p$. The $\dpl$
			responds by choosing  a position $i''_p$ in $\rho_x$ such that $i_p < i''_p < i'_p$. The round ends with the configuration $(i''_p , j''_p)$.  
			If the $\dpl$ cannot choose an $i''_p$ , the game ends and the $\spl$ wins.
		\end{itemize} 	
\end{itemize}
\begin{itemize}
	\item \textbf{Game equivalence:} $(\rho_1, i_1) \approx_{r,I_{\nu}}(\rho_2, j_1)$ iff for every $r$-round 
	$\mathsf{MTL}$ game over the
	words $\rho_1, \rho_2$ starting from the configuration $(i_1, j_1)$, the $\dpl$ always has a winning strategy.
	\item \textbf{Formula equivalence:} $(\rho_1, i_1) \equiv^{\mathsf{MTL}}_{r,I_{\nu}} (\rho_2, j_1)$ iff for every $\mathsf{MTL}$ formula $\phi$ of modal depth $\le r, \rho_1,i_1\models \phi \iff \rho_2,j_1\models \phi$
\end{itemize}
\begin{theorem}
	\label{concur11}
	$(\rho_1,i_1) \approx_{r,I_{\nu}} (\rho_2, j_1)\ iff\ (\rho_1, i_1) \equiv^{\mathsf{MTL}}_{r,I_{\nu}} (\rho_2, j_1)$ \cite{concur11}.
\end{theorem}

\section{Proof of Theorem \ref{ctmtl-game}}
\label{app:tmtlgame}
We prove the result for theorem \ref{ctmtl-game} in this section using structural induction on the number $r$ of rounds. 
We first observe that in the base case $r=0$, the theorem holds: 
If $(\rho_1, i_1) \approx_{0,k,I_{\nu}} (\rho_2, j_1)$, $\dpl$ wins the zero round game.
This is possible iff $isop(i_1,j_1)$. It is then clear that both words satisfy the same formulae of depth 0. The converse is similar.

  Assume the theorem holds for $r=n$ rounds. We will prove the theorem for $n+1$ rounds.

\begin{enumerate}
\item  Assume $(\rho_1, i_1) \approx_{n+1,k,I_{\nu}} (\rho_2, j_1)$. Let us 
consider  $\varphi = \psi \until_{I,\#_\delta \ge w} \phi$.
Assume further that  $\rho_1,i_1\models \varphi$. We need to prove that $\rho_2, j_1 \models\varphi$. 

 \begin{enumerate}
 \item Let us first consider the case when $\spl$ initiates a $\fut_I$ 
 move on $\rho_1$.  Let  $i_1'$ be the point chosen.
 $\dpl$ has to mimic the move by choosing a point $j'_1$.
If $\spl$ ends the round at this point, then by assumption, we know that $\dpl$ 
wins from $(i_2, j_2)=(i'_1, j'_1)$ in an $n$ round game. 
   By induction hypothesis,  we know that  $i_2$ and $j_2$ satisfy the same set of formulae with depth $\le n$. 
   Thus $\rho_2, j_2 \models \phi$. 
\item  Now consider the case that $\spl$ plays a full until round 
 from $(i_1, j_1)$.  Then he chooses a point  $j''_1$ ($j_1< j''_1<j'_1$) in $\dpl$'s word.
  By assumption, duplicator will be able to choose a point $i_1''$ ($i_1< i_1'' <i_1'$) such that he wins 
  the game from $(i_2, j_2)=(i''_1,j''_1)$ in the next $n$ rounds.  By induction hypothesis, all  points 
  between $i_1, i'_1$ as well as between $j_1, j'_1$ satisfy the same set of  formulae of  depth $\leq n$.
   We  know that the depth  of $\psi$ is $\le n$, and 
   all points between $i_1$ and $i'_1$ satisfy $\psi$. 
    Thus all the points strictly between $j_1$ and $j_1'$ satisfy $\psi$.
  Hence, $\rho_2, j_2 \models \psi$.  
  
  \item The  third choice of the $\spl$ is to invoke the $\mathsf{UT}$ move. 
   $\spl$ keeps $w<k$ pebbles  between $i_1$ and $i'_1$. 
   Let $I_1$  be the set of pebbled points in $\spl$'s word. 
   In response,    $\dpl$ also keeps $w$ pebbles in his word between $j_1$ and $j'_1$.
   Let $I_2$ be the set of pebbled positions in $\dpl$'s word. 
   For any choice of a pebbled point $j_2 \in I_2$, $\dpl$ picks some point   
   $i_2 \in I_1$. By assumption, $\dpl$ wins an $n$ round game from 
   this configuration. Hence, all the pebbled positions 
   in both words satisfy the same set of formulae of depth $\leq n$, and 
   in particular $\delta$. Hence, $\rho_2, j_2 \models \delta$.  
\end{enumerate}

Thus, by semantics, points (a), (b) and (c) above give us $\rho_2,j_1 \models  \psi \until_{\{I,\#_\delta \ge w\}} \phi$. 
\item 
We now consider the case 
when the outer most connective is a $\mathsf{C}$  modality. Let $\varphi = \cnt^{\ge w}_{\langle l,u\rangle} \delta$. 
Assume  that  $\rho_1,i_1\models \varphi$. We need to prove that $\rho_2, j_1 \models\varphi$. 
  $\spl$ selects  $w\le k$ points with timestamps in $\langle t_{i_1}+l,t_{i_1}+u \rangle$ that satisfies $\delta$ and keep his pebbles. Let $I_1$ be the set of points pebbled by $\spl$. 
 In response, the $\dpl$ chooses $w$ points with timestamps $\langle t_{j_1}+l,t_{j_1}+u \rangle$. 
 Let $I_2$ be the set of points pebbled by $\dpl$. 
 $\spl$ chooses a point from $I_2$; the duplicator responds with a point in $I_1$. 
 By assumption, for any point  $e_2 \in I_2$ chosen by $\spl$, the $\dpl$ can pick a point in $e_1\in  I_1$ such that 
 from $(e_1, e_2)$, $\dpl$ 
    wins in the next $n$ rounds. By induction hypothesis, $\forall e_2 \in I_2, \exists e_1 \in I_1, (\rho_1,e_1)  \equiv_{n,k} (\rho_2, e_2)$. Note that all the points in $I_1$ satisfy $\delta$. Since  $\delta$ has depth $n$, all the points in $I_2$ also satisfy $\delta$. Thus $\rho_2,j_1$ also satisfies $\cnt^{\ge w}_{\langle l,u\rangle}  \delta$. Hence,  $\rho_2, j_1\models \varphi$. 
\end{enumerate}

 We will now prove the contrapositive. 
If $(\rho_1, i_1) \not\approx_{0,k,I_{\nu}} (\rho_2, j_1)$, then $\neg isop(i_1, j_1)$. 
Then clearly, there is a depth 0 formula that distinguishes $\rho_1, \rho_2$. 
Let us assume the result for  $r=n$
and let  
$(\rho_1, i_1) \not\approx_{n+1,k,I_{\nu}} (\rho_2, j_1)$. 
We construct a  formula of depth $n+1$ that separates $\rho_1, i_1$ and $\rho_2, j_1$.
 Given $\rho_1, \rho_2$ of finite length say $n, m$ respectively, the choice 
  of intervals $I_{\nu}=\{\langle i, j \rangle \mid 0 \leq i \leq max(n,m), 0 \leq j \leq max(n,m)$ or $j=\infty$ and $i \leq j\}$.
  \begin{enumerate}
 \item 
 Assume without loss of generality that $\spl$ chooses $\rho_1,i_1$  to start with and plays 
 a $\fut_I$ by choosing $i'_1>i_1$.  $\dpl$ chooses a point $j'_1>j_1$ in $\rho_2$. 
 If $\spl$ wins from $(i'_1, j'_1)=(i_2,j_2)$,  
  then by induction hypothesis, there is an $n$ depth formula which evaluates to true at $\rho_1,i'_1$ 
  but not at  $\rho_2,j'_1$. Let $Q_x$ be the conjunction of  all depth $n$ formulae that evaluate to true at $\rho_1,x$. 
  For a given $n, k$ and permitted intervals $I_{\nu}$, this conjunction is bounded and finite :
        Thus, if $\spl$ wins after the $\fut_I$ round,  the formula $\fut_I(Q_{j'_1})$ of depth $n+1$ distinguishes the words. 
  
  \item   Suppose that $\spl$ has to play an $\until_I$ round to win the $n+1$ round game. 
  Then $\spl$ picks a point $j''_1$ between $j_1$ and $j'_1$ in $\rho_2$. 
  For any point $i''_1$ between $i_1$ and $i'_1$ picked by the $\dpl$,   the $n$ round game is won by $\spl$. 
   Thus, there exists a point $j_1 < x < j'_1$ 
  and some formula of depth $n$ which distinguishes  
    $x$ from all the points between $i_1$ and $i'_1$.
      Consider the formula  $P = \bigvee \limits_{y\in \{i_1,\ldots i'_1\}} Q_y$. 
 The size of $P$ is bounded since the size of each $Q_y$ is bounded, 
and the number of disjuncts is finite.  
     Hence,  there is a point $j_1''$ ($j_1 <j''_1<j'_1$) such that $\rho_2,j''_1 \nvDash P$. Thus the formula distinguishing $\rho_1, \rho_2$ is  $P \until_I Q_{i_1'}$. 
    
  \item   Suppose that $\spl$  has to play the $\mathsf{UT}$ round  to win the game.
  Assume $\spl$ chose the word  $\rho_1$ 
  and places his $w\leq k$ pebbles at a set of points $I_1$ between $i_1$ and $i'_1$. 
  In response, $\dpl$ keeps his $w\leq k$ pebbles at a set of points $I_2$ 
  between $j_1$ and $j'_1$. 
   $\spl$ picks a point $j''_1\in I_2$, to which $\dpl$  replies by picking  $i''_1\in I_1$. 
  Since $\spl$ wins by assumption, there is a formula of  depth $\leq n$ that distinguishes 
   $j''_1$ from all the points in $I_1$.  
    Now consider the formula $P_{I_1} = \bigvee \limits_{i \in I_1} Q_i$, where
    $Q_i$ be the conjunction of  all depth $n$ formulae that evaluate to true at $\rho_1,i$. 
    For a given $n, k$ and permitted intervals $I_{\nu}$,
    there are a bounded number of $n$ depth formulae; hence the number of different formulae $P_{I_1}$ is bounded. 
      Since $\spl$ wins the game in the next $n$ rounds, $P_{I_1}$
  is true for at least $w$ number of times between $i_1$ and $i'_1$ since 
  it evaluates to true at all points between $i_1$ and $i'_1$. However,  
the number of times   $P_{I_1}$ evaluates to true between 
  $j_1$ and $j'_1$ is $<w$, since it does not evaluate to true at $j''_1$.   
    Hence, $\rho_2,j_1\nvDash P \until_{I,\# P_{I_1} \ge w}Q_{i_1'}$
    where $P = \bigvee \limits_{y\in \{i_1,\ldots i'_1\}} Q_y$ is as defined above. \\

  Similarly, if $\spl$ had pebbled the points $I_2$ between $j_1$ and $j'_1$ in the counting part, 
  then $\dpl$ pebbles the set $I_1$ of points between $i_1$ and $i'_1$. 
  Then $P_{I_2} = \bigvee \limits_{i \in I_2} Q_i$ evaluates to true atleast $w$ times 
  between $j_1$ and $j'_1$, but there is some point $i''_1 \in I_1$ chosen by $\spl$ 
    where $P_{I_2}$ is false. Then 
       the number of times  
    $P_{I_2}$ evaluates to true is $<w$ between $i_1$ and $i'_1$. In this case, 
    $\rho_2, j_1 \nvDash P \until_{I,\# P_{I_2} < w} Q_{i_1'}$.
    
  \item Suppose now that $\spl$ has to play a $\mathsf{C}$ move
  to win the game.   Assume without loss of generality that $\spl$ chooses to play from $\rho_1$.  Let $Q_x$  be the conjunction of all the $n$ depth  formulae having $k$ as the maximum counting constant in the $\mathsf{C,UT}$ modalities that evaluate to true 
at a point $x$. Given that $n, k$ and the possible intervals $I_{\nu}$ are finite, the number of 
formulae $Q_x$ is bounded. Let us consider the case that $\spl$'s first move is a  $\cnt^{\ge k}_{\langle l,u \rangle}$ move. 
$\spl$ pebbles the set $I_1$ of $k$ points in $\langle t_{i_1} + l, {t_{i_1}}+u \rangle$. In response, 
$\dpl$ pebbles the set $I_2$ of $k$ points in $\langle t_{j_1} + l, {t_{j_1}}+u \rangle$.
  $\spl$ picks a point $e_2 \in I_2$, and $\dpl$ replies by choosing $e_1 \in I_1$.
  By assumption, $\dpl$ loses an $n$ round game from $(e_1, e_2)$.
  Hence, by induction, there is a formula $\varphi$ of depth $n$ 
  which will evaluate to false at $e_2$. 
  Consider the formula $Q = \bigvee \limits_{x \in I_1} Q_x$. 
  $Q$ is a formula of depth $n$ having $k$ as the maximum counting constant in its counting modalities 
  since each $Q_x$ is one such. Clearly, $Q$ evaluates to true at all $k$ points 
  of $I_1$; however, the number of points where $Q$ evaluates to true is $<k$ 
  in $I_2$. Hence, $\rho_1, i_1 \models \cnt^{\geq k}_{\langle l, u \rangle}Q$, while 
  $\rho_2, j_1 \models  \cnt^{<k}_{\langle l, u \rangle}Q$.  The formula $\cnt^{\geq k}_{\langle l, u \rangle}Q$
  has depth $n+1$ with max constant $k$ in its counting modalities and distinguishes 
  the two words.

    \end{enumerate}
    Hence, we can show that formula equivalence holds iff $\dpl$ wins in the 
    associated game.

\section{Details of Situation 2 in Proposition \ref{game-details}}
\label{app:situation2}

\flushleft{
\begin{figure}[t]
\begin{tikzpicture}

\foreach \x in {0.25}{
\draw  (\x + 0,0) -- (\x + 8,0);
\draw[dashed] (\x+8,0) -- (\x+10,0);
\draw (\x+10,0) -- (\x + 12,0);
\foreach \y in {0,1,2,3,4}
{
	
\draw (\x+\y*2,-0.25)--(\x+\y*2,0.25);

\node at (\x + \y*2, -0.35) {\tiny \y};

}
\foreach \y in {5,6}
{
	
	\draw (\x+\y*2,-0.25)--(\x+\y*2,0.25);

}
	\node at (\x + 10, -0.35) {\tiny $K$};
		\node at (\x + 12, -0.35) {\tiny $K+1$};



\foreach \y in {0.85,1.1,1.8,3.75,3.9,4.1,4.75,6.7,6.9,7.55,10.5,10.75,11.45}
{\node[fill = magenta,draw = blue,circle,inner sep=1pt,label=below:{}] at (\x+\y,0){};
}
\foreach \y in {1.2,1.25,..., 1.65}
{\node[fill = blue,draw = blue,circle,inner sep=0.5pt,label=below:{}] at (\x+\y,0){};
}

\foreach \y in {4.2,4.25,...,4.65}
	{\node[fill = blue,draw = blue,circle,inner sep=0.25pt,label=below:{}] at (\x+\y,0){};
}
\foreach \y in {7,7.05,...,7.45}
{\node[fill = blue,draw = blue,circle,inner sep=0.25pt,label=below:{}] at (\x+\y,0){};
}
\foreach \y in {10.85,10.9,...,11.25}
{\node[fill = blue,draw = blue,circle,inner sep=0.25pt,label=below:{}] at (\x+\y,0){};
}

\node at (\x + 0.85, -0.25) 
{\tiny {$x_1$} };
\node at (\x + 1.1, 0.15) 
{\tiny {$z_1$} };
\node at (\x + 1.8, -0.25) 
{\tiny {$y_1$} };
\node at (\x + 3.75, -0.25) 
{\tiny {$x_2$} };
\node at (\x + 3.9 , 0.15) 
{\tiny {$z_2$} };

\node at (\x + 4.1, 0.15) 
{\tiny {$e$} };

\node at (\x + 4.75, -0.25) 
{\tiny {$y_2$} };		
\node at (\x + 6.7, -0.25) 
{\tiny {$x_3$} };
\node at (\x + 6.9, 0.15) 
{\tiny {$z_3$} };
\node at (\x + 7.55, -0.25) 
{\tiny {$y_3$} };

%
%
%
%

%
\node at (\x,-0.5){};
			
		\node at (\x + 10.5 , -0.25){\tiny {$x_{K}$} };
		\node at (\x + 10.75, 0.15) 
		{\tiny {$z_{K}$} };
		\node at (\x + 11.5 , -0.25) 
		{\tiny {$y_{K}$} };
}

\end{tikzpicture}
}
\flushleft{
	\begin{tikzpicture}
	\foreach \h in {0.1}
	{
	\foreach \x in {0.25}{
		\draw (\x + 0,0) -- (\x + 8,0);
		\draw[dashed] (\x+8,0) -- (\x+10,0);
		\draw (\x+10,0) -- (\x + 12,0);
		\foreach \y in {0,1,2,3,4}
		{
			
			\draw (\x+\y*2,-0.25)--(\x+\y*2,0.25);
			
			\node at (\x + \y*2, -0.35) {\tiny \y};
			
		}
		\foreach \y in {5,6}
		{
			
			\draw (\x+\y*2,-0.25)--(\x+\y*2,0.25);

		}
		\node at (\x + 10, -0.35) {\tiny $K$};
		\node at (\x + 12, -0.35) {\tiny $K+1$};


			 {0.85,1.1,1.8,3.79,4,4.75,6.7,6.9,7.55,10.5,10.75,11.45}

			\node at (\x + 0.85 -\h, -0.25) 
				{\tiny {$x_1'$} };
			\node at (\x + 1.1 -\h, 0.15) 
			{\tiny {$z_1'$} };
			\node at (\x + 1.8 -\h, -0.25) 
			{\tiny {$y_1'$} };
			\node at (\x + 10.5, -0.25) 
				{\tiny {$x_{K}'$} };
				\node at (\x + 10.75, 0.21) 
				{\tiny {$z_{K}'$} };
				\node at (\x + 11.5 -\h, -0.25) 
				{\tiny {$y_{K}'$} };
			
			\node at (\x + 3.9 -\h, -0.25)  {\tiny {$x_2'$} };
			\node at (\x + 4.3 -\h, 0.21) 
			{\tiny {$z_2'$} };
			\node at (\x + 4.75 -\h, -0.25) 
			{\tiny {$y_2'$} };		
			\node at (\x + 6.7 -\h, -0.25) 
			{\tiny {$x_3'$} };
			\node at (\x + 6.9 -\h, 0.15) 
			{\tiny {$z_3'$} };
			\node at (\x + 7.55 -\h, -0.25) 
			{\tiny {$y_3'$} };	
						
		\foreach \y in {0.85,1.1,1.75,3.95,4.2,4.75,6.7,6.9,7.55,10.5,10.75,11.45}
		{\node[fill = magenta,draw = blue,circle,inner sep=1pt,label=below:{}] at (\x+\y - \h ,0){};
		}
		\foreach \y in {1.2,1.25,..., 1.65}
		{\node[fill = blue,draw = blue,circle,inner sep=0.5pt,label=below:{}] at (\x+\y - \h,0){};
		}
		
		\foreach \y in {4.3,4.35,...,4.65}
		{\node[fill = blue,draw = blue,circle,inner sep=0.25pt,label=below:{}] at (\x+\y-\h,0){};
		}
		\foreach \y in {7,7.05,...,7.45}
		{\node[fill = blue,draw = blue,circle,inner sep=0.25pt,label=below:{}] at (\x+\y-\h,0){};
		}
		\foreach \y in {10.95,11,...,11.45}
		{\node[fill = blue,draw = blue,circle,inner sep=0.25pt,label=below:{}] at (\x+\y-\h-\h,0){};
		}
		%
		%
		%
		%
		
		%
		}
}
	\end{tikzpicture}
\caption{Words showing $\mathsf{CMTL}-\mathsf{TMTL} \neq \emptyset$}
\label{express:fig1} 
\end{figure}
}
\noindent{\bf Situation 2}: 
Starting from $(i_1, j_1)$ with time stamps (0,0), 
if the $\spl$ chooses a $\until_{(0,1)\#_a \sim c}$ move and 
lands up at some point between $x_1$ and $y_1$, $\dpl$ will play copy-cat
and achieve an identical configuration.  
 Consider the case when 
$\spl$ lands up at $y_1$\footnote{The argument when $\spl$ lands up at $x_1$ or a point in between $x_1, y_1$ is exactly the same}. In response,  $\dpl$ moves to   
$y'_1$. From configuration $(i_2, j_2)$ with time stamps $(y_1, y'_1)$, 
consider the case when $\spl$ initiates a $\until_{(1,2)\#_a \sim c}$ and moves 
 to $z_2=1.8+\epsilon <2$. In response, $\dpl$ 
 moves to the point $z'_2=2.1>2$. 
 A pebble is kept at the inbetween 
positions $x_2, x'_2$ respectively. 
If $\spl$ chooses to pick the pebble in 
$\dpl$'s word, then we obtain the configuration $(i_3, j_3)$ with time stamps  
$(x_2, x'_2)$.  If $\spl$ does not get into the counting part/until part, the 
configuration obtained has time stamps  $(z_2, z'_2)$, with the lag 
of one segment ($seg(i_3)=1$, $seg(j_3)=2$). 
\begin{itemize}
\item  Assume we have the  configuration $(i_3,j_3)$ with time stamps $(x_2, x'_2)$. 
We know that  $y_j-x_2, y'_j-x'_2, z'_j-x'_2, z_j-x_2  \in (j-1, j)$  
and $x'_j-x'_2, x_j-x_2 \in (j-2, j-1)$ for $j \geq 3$, and 
$y'_2-x'_2, y_2-x_2, z_2-x_2, z'_2-x'_2 \in (0,1)$. 
   Thus,  from $(x_2, x'_2)$, the possible configuration 
   obtained is $(i'_3, j'_3)=(k_j, k'_j)$ with $k \in \{x,y,z\}$ and $j \geq 3$ or 
   $(z_2, z'_2)$ or $(y_2, y'_2)$. 
  In the case of $(z_2, z'_2)$, there are no inbetween positions for pebbling. 
  In all the other cases, as long as $\spl$ does not keep a pebble 
  on $z_2$, we will either obtain $(i_4, j_4)=  (k_j, k'_j)$ with $k \in \{x,y,z\}$ and $j \geq 3$
or $(y_2, y'_2)$. If $\spl$ keeps just one pebble, and that too on $z_2$, then $\dpl$ 
will keep his only pebble at $z'_2$ obtaining $(i_4, j_4)=(z_2, z'_2)$.  In all the cases other than obtaining 
$(i_4, j_4)$ with time stamps $(z_2, z'_2)$, there is no segment lag. In fact, all these cases give an identical configuration with same time stamps, from where $\dpl$ can easily win. Lets hence look at 
the case of $(z_2, z'_2)$.  
 \item Consider the configuration $(i_3,j_3)$ with time stamps $(z_2, z'_2)$. 
 In this case, there is a lag of one segment.
 \begin{itemize}
 \item[(a)]  If $\spl$ chooses to move from $z'_2$ to 
  $x'_3$, then $\dpl$ can move to $x_3$ from $z_2$, since $x_3-z_2, x'_3-z'_2 \in (1,2)$ and 
  obtain a configuration $(i'_3, j'_3)$ with same time stamps $(x_3, x'_3)$. 
  If $\spl$ does not pebble the points in between, we obtain an identical configuration 
  with time stamps $(x_3, x'_3)$, from where it is easy to see that $\dpl$ wins. If $\spl$ pebbles 
  points between $z_2$ and $x_3$, the interesting situation is when he pebbles only $e$; in this case, $\dpl$'s best choice is to pebble the point (say $e'$) right after $z'_2$ since $e'-e=\kappa$. The configuration with time stamps $(e,e')$ is as good as an  identical configuration.
  
 \item[(b)]   Lets see the case when  $\spl$ moves to $z'_3$ or $y'_3$ from $z'_2$.  
 Then $\dpl$'s best choice is to move to $x_3$ from $z_2$, since he cannot move to 
 $z_3, y_3$(
 $z'_3-z'_2  \in (1,2), y'_3-z'_2 \in 
  (1,2)$, but $z_3-z_2, y_3-z_2  \in (2,3)$).   This gives the configuration $(i'_3, j'_3)$ with time stamps 
  $(x_3, y'_3)$ or $(x_3, z'_3)$, with no lag in the segments of $\rho_1, \rho_2$.
  If $\spl$ pebbles the positions inbetween $z'_2$ and $y'_3$ (or $z'_3$), 
  then $\dpl$ places his pebbles among the bunch of points between $e$ and $y_2$. The resultant 
  configuration is $(i_4, j_4)$ with the following interesting possibilties:
  \begin{enumerate}
 \item[(c)]  $e< t_{i_4} < y_2$   and 
  $z'_3 < t_{j_4} < y'_3$. From 
  $(i_4, j_4)$, if $\spl$ moves to any point $k_j$ (or $k'_{j+1}$)
  for $k \in \{x,y,z\}$ and $j \geq 3$, $\dpl$ can move into $k'_{j+1}$ (or $k_j$) since 
  for any interval $I$, $k_j-t_{i_4} \in I$ iff $k'_{j+1}-t_{j_4} \in I$. This 
  results in future configurations of the kind having 
  time stamps $(k_j, k'_{j+1})$ for $j \geq 3, k \in \{x,y,z\}$, and 
  a segment lag of 1.
    \item[(d)] $e< t_{i_4} < y_2$   and 
  $t_{j_4}=x'_3$. From $(i_4, j_4)$, 
  the  reachable configurations $(i'_4, j'_4)$ are those where 
  $e < t_{i'_4}<y_2$, $t_{j'_4}=z'_3$, when both players move to the next point ($i'_4=i_4+1, j'_4=j_4+1$) 
  or  $e < t_{i'_4}<y_2$, $z'_3 < t_{j'_4}<y'_3$ (case above) or 
  with time stamps $(y_2, y'_3)$, or     
  $(k_j, k'_{j+1})$ for $j \geq 3, k \in \{x,y,z\}$.
  All these  
  result in future configurations of the kind having 
  time stamps $(k_j, k'_{j+1})$ for $j \geq 3, k \in \{x,y,z\}$, and 
  a segment lag of 1. 
  
  \item[(e)] $e<t_{i_4}<y_2$ and 
  $z'_2<t_{j_4}<x'_3$. This is like an identical configuration, and 
  from here, $\dpl$ can stay in the same segment as $\spl$ in all future moves, obtaining 
  almost identical configurations. 
      \end{enumerate}
  \end{itemize}
      \end{itemize}

\section{Proof of Lemma \ref{lem1}}
\label{app:lem1}
\subsection*{$\mathsf{MTL} \subseteq \mathsf{C^{(0,1)}MTL}$}
The containment of $\mtl$ in $\mathsf{C^{(0,1)}MTL}$ is clear since $\mathsf{C^{(0,1)}MTL}$ has all the modalities of 
 $\mathsf{MTL}$. 
 We show strict containment  by considering the formula $\varphi=\cnt^{=2}_{(0,1)}a \in \mathsf{C^{(0,1)}MTL}$.
 We show that for any choice $n$ of rounds, we can find two timed words $\rho_1, \rho_2$ such that 
 $\rho_1 \models \varphi, \rho_2 \nvDash \varphi$, but 
 $\rho_1 \equiv_n^{\mathsf{MTL}} \rho_2$. 
 
 Consider the timed words  $\rho_1=(a,0)(a,0.5)(a,0.6)W$ and $\rho_2=(a,0)(a,0.5)W$ where $W$ is  $(a,1.1)(a,1.1+\delta)(a,1.1+2 \delta) \dots (a,1.1+n \delta)$, where $\delta < < \frac{1}{n}$  
 is some small constant such that $1.1+n \delta < 1.2$.
    Clearly, $\rho_1 \models  \varphi$ and $\rho_2 \nvDash \varphi$.
 Since the words are identical from time 1.1 onwards, the interesting parts of the game are in the interval (0,1).
 
 \begin{proposition}
 \label{app:exp1}
 In any round $p$ of the $\mathsf{MTL}$ game, 
 $\dpl$ can always ensure an identical configuration $(i_p, j_p)$ ($i_p=j_p$)
  or ensure that $|i_p-j_p| \leq 1$.
 If $i_p-j_p=1$ and $i_p \geq 3$, then 
 for all $q > p$, $\dpl$ can ensure that $0\leq i_q-j_q \leq 1$.
Further, the number of positions to the right of any word 
during the $p$th round will be either same, 
or $n+3-p$ and $n+2-p$ respectively for $\rho_1, \rho_2$.
 \end{proposition}
 \begin{proof}
  The starting configuration is $(i_1,j_1)$,  the starting positions 
  of the two words.   Assume $\spl$ chooses the word 
 $\rho_1$, while $\dpl$ chooses $\rho_2$. Choosing the interval $I=(0,1)$, $\spl$ invokes 
 a $\until_I$ move and chooses one of the $a$'s in $(0,1)$. In response, 
 $\dpl$ chooses the only $a$ at 0.5 in $(0,1)$ in $\rho_2$. 
 The possible configurations are those with time stamps (0.5,0.5) or (0.6, 0.5). 
 The configuration with time stamps (0.6, 0.5) is such that $i_2-j_2=3-2=1$, 
    both words have exactly the same 
 symbols in the future, at the same time points. 
 Thus, $\dpl$ can achieve a configuration with identical time stamps, 
 preserving the lag of one position. 
  
  Let us now look at the configuration $(i_2, j_2)$ with time stamps (0.5, 0.5).  
Assume $\spl$ continues to play in $\rho_1$, and 
   chooses the $a$ at 0.6 by a $\until_{(0,1)}$ move. 
 In this case, $\dpl$ will choose the $a$ at 1.1, obtaining the configuration 
with time stamps  (0.6,1.1). The configuration $(i_3,j_3)$ with time stamps (0.6,1.1) is such that
  $i_3=j_3$.   Note that from (0.6,1.1), $\dpl$ can always ensure 
 an identical configuration $i_p=j_p, p \geq 3$ ($\dpl$ always moves the same number of positions 
 as the $\spl$) 
 or ensure a lag of one position (in this case, $\spl$ 
moves ahead by more than one position  and $\dpl$ also chooses the position with the same time stamp).
Since the number of positions in $\rho_1$ is $n+3$ and that in $\rho_1$ is $n+2$,
the number of positions to the right of any word 
during the $p$th round will be either same, 
or $n+3-p$ and $n+2-p$ respectively.

If $\spl$ starts playing from $\rho_2$, and 
chooses the $a$ at 0.5 using a $\until_{(0,1)}$ move,
 then $\dpl$ also chooses the $a$ at 0.5 in $\rho_1$. 
  If $\spl$ swaps the words at the end of this move, then $\dpl$ can achieve identical configurations 
  for the rest of the game; otherwise, he can ensure a lag of atmost one position as seen above. 
   \end{proof}
 
\subsection*{$\mathsf{C^{(0,1)}MTL} \subseteq \mathsf{C^{0}MTL}$}

The containment of $\mathsf{C^{(0,1)}MTL}$ in $\mathsf{C^0MTL}$ follows from the fact that $\mathsf{C^{(0,1)}MTL} \subseteq \mathsf{C^0MTL}$.
To show the strict containment, 
consider the formula $\varphi=\cnt_{(0,2)}^{\geq 2} a \in \mathsf{C^0MTL}$. 
We show that for any choice of $n$ rounds and $k$ pebbles, 
we can find two words $\rho_1, \rho_2$ such that 
 $\rho_1 \models \varphi, \rho_2 \nvDash \varphi$, but 
 $\rho_1 \equiv_{n,k}^{\mathsf{C^{(0,1)}MTL}} \rho_2$. 
 
Consider the words 
$\rho_1=(a,0)(a,1.8)(a,1.9)W$ and $\rho_2=(a,0)(a,1.9)W$ where $W$ 
is $(a,2.1)(a,2.1+\delta) \dots, (a, 2.1+nk \delta)$ where $\delta < < \frac{1}{2nk}$ 
such that 2.1+$nk \delta < 2.2$. Clearly, $\rho_1 \models \varphi$ while $\rho_2 \nvDash \varphi$. 

 \begin{proposition}
 \label{app:exp2}
 In any round $p$ of the $\mathsf{C^{(0,1)}MTL}$ game, 
 $\dpl$ can always ensure an identical configuration $(i_p, j_p)$
  or ensure that $0 \leq i_p-j_p \leq 1$.
 If $i_p-j_p=1$ and $j_p \geq 2$, then 
 for all $q > p$, $\dpl$ can ensure that $0 \leq i_q-j_q \leq 1$.
Further, the number of positions to the right of any word 
during the $p$th round will be either same, 
or $nk+4-p$ and $nk+3-p$ respectively.
     \end{proposition}
\begin{proof}
The initial configuration is $(i_1,j_1)$ with time stamps (0,0). Assume $\spl$ 
 picks $\rho_2$ while $\dpl$ chooses $\rho_1$.
 The first move cannot be a counting move, since 
 there are no points in (0,1) in both the words.
$\spl$ invokes an $\until_{(1,2)}$ move and comes to  1.9 
in $\rho_2$, while duplicator comes to 1.8 in $\rho_1$ (note that if $\spl$ comes to a point $>$1.9, $\dpl$ 
comes to the point with the same time stamp).
There are no inbetween points to be chosen in an until part, so the 
configuration is $(i_2,j_2)$ with time stamps (1.8, 1.9). This configuration is such that 
$i_2=j_2$, and the number of positions 
on the right are respectively $nk+2$ and $nk+1$.
Lets consider the 2cnd round starting with $(i_2, j_2)$
having time stamps (1.8, 1.9). 
    Assume $\spl$ chose $\rho_1$
 and a point $i'_2$ with  $t_{i'_2}>1.8$. 
 In response, $\dpl$ will choose 
 a point $j'_2$ with   $t_{j'_2}> 1.9$. 
If $t_{i'_2}$ is 1.9, then  $t_{j'_2}$ is 2.1 
with $i'_2=j'_2$.  If  $t_{i'_2}>1.9$, 
  then  $t_{j'_2}> 2.1$ and $\dpl$ can ensure  
  $t_{i'_2}=t_{j'_2}$. 
   When $\spl$ pebbles positions between $i_2$ and $i'_2$, 
$\dpl$ can pebble in such a way that either $i_3=j_3$, or 
$t_{i_3}=t_{j_3}$, and $i_3-j_3=1$. 
When $t_{i_3}=t_{j_3}$, the number of points to the right is the same in both the words.
In the case $t_{i'_2}$ =1.9, and  $t_{j'_2}$= 2.1,  
we obtain  $i_3=j_3$ with  $nk+1$ and $nk$ positions respectively  
on the right in $\rho_1, \rho_2$.

Assume that at the start of the $p$th round, we have 
the configuration $(i_p, j_p)$
with $i_p=j_p$, and having respectively $nk+4-p$ and 
$nk+3-p$ positions to the right in $\rho_1, \rho_2$. 
 Assume further that $\spl$ chooses $i'_p > i_p$, and $\dpl$ chooses $j'_p > j_p$ 
as part of the $p$th round's play.
If $i'_p > i_p+1$, then $\dpl$ chooses 
$j'_p$ such that $t_{i'_p}=t_{j'_p}$, 
from where the number of positions on the right 
is the same in both words. 
Even if $\spl$ decides to play a full until round by choosing a point $j''_p$ 
in between $j_p$ and $j'_p$, $\dpl$ can always choose $i''_p$ having the same time stamp 
as $i''_p$. If $i'_p=i_p+1$, then $\dpl$ has to choose 
$j'_p=j_p+1$, and in this case, the lag of one position continues to the next configuration
with $i_{p+1}=i'_p$ and $j_{p+1}=j'_p$. 
We then have respectively $nk+3-p$ and 
$nk+4-p$ positions to the right in $\rho_1, \rho_2$.
In the case  $i_p-j_p=1$, we have  $t_{i_p}=t_{j_p}$. In this case, $\dpl$ can always ensure 
$i_{p+1}-j_{p+1}=1$, and 
the number of positions to the right 
are same in both the words.   
   \end{proof}

  \section{Proof of Lemma \ref{cmtlless}}
  \label{app:cmtl-less}
  Consider the formula $\varphi=\fut_{(0,1) \#_a \geq 3}b \in \mathsf{TMTL}$.
We show that for any choice of $n$ rounds and $k$ pebbles, 
we can find two words $\rho_1, \rho_2$ such that 
 $\rho_2 \models \varphi, \rho_1 \nvDash \varphi$, but 
 $\rho_1 \equiv_{n,k}^{\mathsf{CMTL}} \rho_2$.

Let $l \in \mathbb{N}$ be the maximum constant used by $\spl$ 
in the set $I_{\nu}$ of permissible intervals. Let $K=nlk+nl$, $\epsilon < \frac{1}{(10)^{10nk}}$ and 
$\kappa=nk\epsilon$.  
Let $0.1>>nk\delta$ and $\delta > \kappa$. \\

\noindent{\it Design of Words}
\begin{enumerate}
\item 
Consider the word $\rho_1$ of length $K+1$. 
Each unit interval $(i, i+1)$ in $\rho_1$, 
$0 \leq i \leq K$ is composed of 3 blocks $A_i, B_i, C_i$, one after the other. 
\begin{itemize}
 \item Block $A_i$ 
has the points  $x_{i1}=i+0.1+\epsilon+i\delta,  y_{i1}=i+0.1+ \kappa+i\delta,  z_{i1}=i+ 0.2+i\delta$
 \item Block $B_i$ has the points  $x_{i2}=i+ 0.3+\epsilon+i\delta, y_{i2}=i+0.3+ \kappa+i\delta,  z_{i2}=i+0.4+i\delta$, and 
 \item  Block $C_i$ has the points $x_{i3}=i+0.5+\epsilon+i\delta, y_{i3}=i+0.5+ \kappa+i\delta, z_{i3}=0.9+i\delta$. 
 \item  Moreover, there are $p>>2nlk$ points in between $x_{ij}$ and $y_{ij}$
  for $1 \leq j \leq 3$.    $\sigma_{z_{ij}}=a$ for all $i, j$, 
the  points  $x_{ij}, y_{ij}$ as well as all points between them are marked $b$.
\item It can be seen that the blocks $A_i, B_i, C_i$ 
shift to the right by $\delta$, as $i$ increases from 0 to $K$. 
\end{itemize}
\item 
The word $\rho_2$ has also length $K+1$. 
Each unit interval 
$(i, i+1)$  (except the last one) in $\rho_2$ is composed of 4 blocks 
$A'_i, B'_i, C'_i, D'_i$ one after the other. 
\begin{itemize}
\item Block $A'_i$ 
has the points  
 $x'_{i1}=i+0.1+\epsilon+i\delta,  y'_{i1}= i+0.1+\kappa+i\delta, z'_{i1}=i+0.2+i\delta$, 
 \item Block $B'_i$ has the points  
 $x'_{i2} =i+0.3+\epsilon+i\delta, y'_{i2}=i+0.3+\kappa+i\delta,   z'_{i2}=i+0.4+i\delta$ 
 \item Block $C'_i$  has the points $x'_{i3}=i+0.5+\epsilon+i\delta$, 
  $y'_{i3}= i+0.5+\kappa+i\delta, z'_{i3}=0.9+i\delta$,  
 \item Block $D'_i$ has the points  
 $x'_{i4}=i+0.99+\epsilon+i\delta, y'_{i4}=i+0.99+\kappa+i\delta$. 
 \item    Moreover, there are $p>>2nlk$ points in between $x_{ij}$ and $y_{ij}$
  for $1 \leq j \leq 4$.
 $\sigma_{z'_{ij}}=a$ for all $i, j$, and
the  points  $x_{ij}, y_{ij}$ as well as all points between them are marked $b$.
\item It can be seen that the blocks $A'_i, B'_i, C'_i$ and $D'_i$ 
shift to the right by $\delta$, as $i$ increases from 0 to $K$.
\item The last unit interval $(K-1, K)$ has only 3 blocks of $b$'s 
starting respectively at  $x'_{K-1~1}, x'_{K-1~2}$ and
$x'_{K-1~3}$ and ending at  $y'_{K-1~1}, y'_{K-1~2}$ and
$y'_{K-1~3}$. The 3 $a$'s occur at  $z_{K-1~1},z_{K-1~2}$ and
$z_{K-1~3}$.
\end{itemize}
\end{enumerate}
\flushleft{
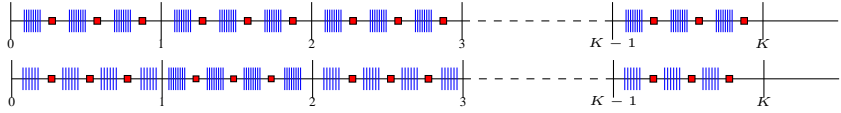
\begin{figure}[th]
\begin{tikzpicture}

\foreach \x in {0.25}{
\draw  (\x + 0,0) -- (\x + 6,0);
\draw[dashed] (\x+6,0) -- (\x+8,0);
\draw (\x+8,0) -- (\x + 11,0);
\foreach \y in {0,1,2,3}
{
	
\draw (\x+\y*2,-0.25)--(\x+\y*2,0.25);

\node at (\x + \y*2, -0.3) {\tiny \y};

}
\foreach \y in {4,5}
{
	
	\draw (\x+\y*2,-0.25)--(\x+\y*2,0.25);

}
	\node at (\x + 8, -.3) {\tiny $K-1$};
		\node at (\x + 10, -.3) {\tiny $K$};

\foreach \u in {0,2,4,8}
{	

	\foreach \y in {0.55,1.15,1.75}
	\node[fill = red,draw = black,rectangle,inner sep=1.25pt,label=below:{}] at (\x+\y+\u,0){};

	\foreach \z in {0.1,0.7,1.3}
{	\foreach \y in {0.08,0.11,...,0.3}
	{\draw[blue] (\x+\y+\z+\u,-0.15)--(\x+\y+\z+\u,0.15);
	} 

}
}
}
\end{tikzpicture}
}

	\flushleft{
		\begin{tikzpicture}
	
		\foreach \x in {0.25}{
			\draw  (\x + 0,0) -- (\x + 6,0);
			\draw[dashed] (\x+6,0) -- (\x+8,0);
			\draw (\x+8,0) -- (\x + 11,0);
			\foreach \y in {0,1,2,3}
			{
				
				\draw (\x+\y*2,-0.25)--(\x+\y*2,0.25);
				
				\node at (\x + \y*2, -0.3) {\tiny \y};
				
			}
			\foreach \y in {4,5}
			{
				
				\draw (\x+\y*2,-0.25)--(\x+\y*2,0.25);

			}
			\node at (\x +8, -0.3) {\tiny $K-1$};
			\node at (\x + 10, -0.3) {\tiny $K$};

			\foreach \u in {0,4}
			{	
				
				\foreach \y in {0.53,1.04,1.54}
				\node[fill = red,draw = black,rectangle,inner sep=1.25pt,label=below:{}] at (\x+\y+\u,0){};
				
				\foreach \z in {0.08,0.61,1.12,1.65}
				{	\foreach \y in {0.07,0.11,...,0.3}
					{\draw[blue] (\x+\y+\z+\u,-0.15)--(\x+\y+\z+\u,0.15);
					} 
					
				}
			}
			\foreach \u in {8}
			{	
				
				\foreach \y in {0.53,1.04,1.54}
				\node[fill = red,draw = black,rectangle,inner sep=1.25pt,label=below:{}] at (\x+\y+\u,0){};
				
				\foreach \z in {0.08,0.61,1.12}
				{	\foreach \y in {0.07,0.11,...,0.3}
					{\draw[blue] (\x+\y+\z+\u,-0.15)--(\x+\y+\z+\u,0.15);
					} 
					
				}
			}

		\foreach \u in {1.9}
		{	
			
			\foreach \y in {0.55,1.05,1.55}
			\node[fill = red,draw = black,rectangle,inner sep=1.1pt,label=below:{}] at (\x+\y+\u,0){};
			
			\foreach \z in {0.11,0.62,1.11,1.65}
			{	\foreach \y in {0.08,0.11,...,0.3}
				{\draw[blue] (\x+\y+\z+\u,-0.15)--(\x+\y+\z+\u,0.15);
				} 
				
			}
		}
		
	}
	\end{tikzpicture}
\caption{The red square represents $a$, the bunch of blue lines represents a bunch of $b$'s. 
There are 3 $a$'s in each unit interval of both $\rho_1$ and $\rho_2$. The difference is that 
$\rho_1$ has 3 blocks of $b$'s in each unit interval, while $\rho_2$ has 4 blocks of $b$'s in each unit interval except the last.  
Clearly, $\rho_2 \models \varphi, \rho_1 \nvDash \varphi$. The time stamp of the three $a$'s are $z_{i1}$, $z_{i2}$ and $z_{i3}$ respectively
 in the $i$th unit interval of $\rho_1$. Likewise, the $j$th bunch of $b$'s in the $i$th unit interval begins 
 with time stamp $x_{ij}$ and ends with time stamp $y_{ij}$. In $\rho_2$ we have $z'_{ij}, x'_{ij}$ and $y'_{ij}$. }
		\label{app:fig2}
				\end{figure}
				}

Clearly, $\rho_1, \rho_2$ have three occurrences of $a$ in each unit interval; however
while $\rho_1$ has 3 blocks of $b$'s in each unit interval with an $a$ after each $b$ block, $\rho_2$ has 
4 blocks of $b$'s in each unit interval, with the 3 $a$'s 
in between the $b$ blocks. Thus, $\rho_2 \models \varphi$ while 
$\rho_1 \nvDash \varphi$. \\

\noindent{\it Topological Similarity of $\rho_1, \rho_2$}: 
Note that for any $i$, 
the catenation of the blocks $D'_i$ and $A'_{i+1}$
is topologically similar to the block $A_{i+1}$:
(i) both have a sufficiently long sequence of $b$'s followed by an $a$; $D'_iA'_{i+1}$ has
$2p+2$ $b$'s followed by an $a$, while $A_{i+1}$ has $p+2$ $b$'s followed by an $a$. Since $p>>2nlk$, and 
the number of rounds is $n$, a bunch of $2p+2$ $b$'s is as good as a bunch of $p+2$ $b$'s. 
(ii) Map $A'_1$, $B'_1$ and $C'_1$ respectively to $A_1, B_1$ and $C_1$; 
map $D'_iA'_{i+1}$ to $A_{i+1}$, $B'_{i+1}$ to $B_{i+1}$ and $C'_{i+1}$ to $C_{i+1}$ for 
$i \geq 1$.\\
\noindent{\it Segmented View of $\rho_1, \rho_2$}: 
We will refer to the unit interval $(i, i+1)$ for $i \geq 0$ in either word
as the $(i+1)$th segment. 
Thus, both the words have $K$ segments numbered $1, \dots, K$. 
For a position $i_p \in dom(\rho_1) \cup dom(\rho_2),$
$seg(i_p)$ represents the segment containing $t_{i_p}$. For instance, 
if $t_{i_p}=5.3$, then the position $i_p$ is contained in segment 6, or $seg(i_p)=6$.\\

\noindent{\it Copy-cat strategy}: Consider the $p$th round of the game with initial configuration $(i_p, j_p)$. 
If $\dpl$ can ensure that $seg(i_{p+1}){-}seg(i_p)$=$seg(j_{p+1}){-}seg(j_p)$, then 
we say that $\dpl$ has adopted a {\it copy-cat} strategy in the $p$th round.

We will now play a $(n,k)$-$\mathsf{CMTL}$ game and show that 
$\dpl$ wins. It is easy to see that $\dpl$ can respond to any of the 
$\until_I$ moves of $\spl$ by the choice  of the words. 
\begin{proposition}
For an $n$ round $\mathsf{CMTL}$ game over the words $\rho_1, \rho_2$, 
the $\dpl$ always has a winning strategy such that for any $1 \leq p \leq n$, if 
$(i_p, j_p)$ is the initial configuration of the $p$th round, then 
$|seg(i_p)-seg(j_p)| \leq 1$. Moreover, when 
$|seg(i_p)-seg(j_p)| =1$, then there are atleast $(n-p)(l+1)$ segments to the right 
on each word after $p$ rounds, for all $1 \leq p \leq n$.
    \end{proposition}
\begin{proof}
Assume that $\spl$ initiates a $\cnt_I^{\geq k}$ move on $\rho_2$. 
Then $\spl$ places $k$ pebbles on $k$ positions of $\rho_2$ 
in the interval $I$ and in response, $\dpl$ pebbles $k$ positions 
in the same interval $I$ of $\rho_1$. 
\begin{enumerate}
\item If $\spl$  does not keep any pebble 
on the last $b$ block in any of the unit intervals spanning $I$,  then $\dpl$
puts his pebbles exactly at the same positions as $\spl$, and obtains an identical 
configuration.
\item \emph{Choice of Pebbling}: Assume 
that we have an identical configuration $(i_p, j_p)$. 
Let us look at $\spl$'s placement of pebbles on some unit interval say $(g,g+1)$. 
Assume that  $\spl$ keeps some (say $l$) pebbles on the last $b$ block (say $D'_g$), and 
$l'$ pebbles on the remaining 3 blocks $A'_g, B'_g$ and $C'_g$ 
of the unit interval  $(g,g+1)$.  In response, 
$\dpl$ places $l'$ of his pebbles at identical positions on $A_g, B_g$ and $C_g$, and 
places $l$ pebbles on $A_g$.
$\dpl$ will place these $l$ pebbles in the first half of $A_g$. 
  Note that since the number of positions in each block 
is $2nk >>l$, this is possible. This way, $\dpl$ keeps his $k$ pebbles on the same unit intervals as $\spl$.
If $\spl$ picks a pebble in $\dpl$'s word  from any $B$ or $C$ block, then $\dpl$ will pick the same pebble 
from $\spl$'s word. If $\spl$ picks a pebble from  $A_g$, then 
there are two possibilities: (i) either this pebble corresponds 
to a pebble kept by $\spl$ on $A'_g$, or (ii) this is one of the $l'$ pebbles kept by $\dpl$  
on $A'_g$ in response to $\spl$'s $l'$ pebbles on $D'_g$. In case of (i), $\dpl$ simply picks the corresponding pebble 
from $A'_g$, obtaining an identical configuration, while in case (ii), $\dpl$ picks the corresponding pebble from $D'_g$.
 This gives a configuration $(i_{p+1}, j_{p+1})$ with $i_{p+1}$ being a position in $A_g$, and 
 $j_{p+1}$ in $D'_g$. So far, there is no lag in the segments, $seg(i_{p+1})=g+1=seg(j_{p+1})$.  
 
   
 \item  Consider  any $\until_I$ move 
or $\cnt_I^{\geq k}$ move that $\spl$ launches on either 
of the words from $(i_{p+1}, j_{p+1})$. 
Recall that $seg(i_{p+1})=g+1=seg(j_{p+1})$, $i_{p+1} \in A_g$, $j_{p+1} \in D'_g$.
\begin{itemize}
\item[(a)]
If $\spl$  moves to some point in  
$A'_h$  (in segment $h+1$), then $\dpl$ will move to some point 
in $A_{h-1}$ (in segment $h$). This is possible since for any interval $I$,  $y'_{h1}-x'_{g4} \in I$ iff 
$y_{h-1~1}-x_{g1} \in I$. 
\item[(b)] 
If $\spl$  moves to some point in  
$B'_h$  (in segment $h+1$), then $\dpl$ will move to some point 
in $B_{h-1}$ (in segment $h$). This is possible since for any interval $I$,  $y'_{h2}-x'_{g4} \in I$ iff 
$y_{h-1~ 2}-x_{g1} \in I$.
\item[(c)] 
If $\spl$  moves to some point in  
$C'_h$  (in segment $h+1$), then $\dpl$ will move to some point 
in $C_{h-1}$ (in segment $h$). This is possible since for any interval $I$,  $y'_{h3}-x'_{g4} \in I$ iff 
$y_{h-1~ 3}-x_{g1} \in I$.
\item[(d)] If $\spl$ moves to some point in  
$D'_h$  (in segment $h+1$), then $\dpl$ will move to some point 
in $A_{h}$ (in segment $h+1$). This is possible since for any interval $I$,  $y'_{h4}-x'_{g4} \in I$ iff 
$y_{h 1}-x_{g1} \in I$.
\end{itemize}
Cases (a)-(c) creates a lag of one segment between the two words, while (d) 
is similar to $(i_{p+1}, j_{p+1})$. From (d), we can achieve any one of cases (a)-(d) listed above. 
Let us hence look at cases (a)-(c), to  
understand the potential future configurations. \\

In cases (a)-(c), 
    when $\spl$ pebbles $k$ positions between 
segments $g+1$ and $h+1$ in $\rho_2$, $\dpl$ pebbles 
$k$ positions between segments $g+1$ and $h$ in $\rho_1$. The choice of pebbling 
is as described in item 2 above:
\begin{itemize}
\item[(i)]  whenever $\spl$ pebbles positions in block $D'_s$
 of segment $g+1 \leq s+1 < h+1$
$\dpl$ pebbles positions in the first half of block $A_{s}$.
If $\spl$ picks one of these pebbles from $A_s$, then $\dpl$ chooses the corresponding pebble
from $D'_s$, thereby obtaining a configuration $(i_q, j_q)$ with $seg(i_q)=seg(j_q)=s+1$, 
with $i_q \in A_s, j_q \in D'_s$. This configuration is exactly same as the one described in 3(d) above. 
 %
   \item[(ii)] whenever $\spl$ places pebbles on $A'_{s}, B'_{s}$ and $C'_{s}$, 
 $\dpl$ places his pebbles  
on $A_s, B_s$ and $C_s$ respectively, for $s < h+1$.
If $\spl$ picks one of these pebbles from $A_s, B_s$ or $C_s$, then $\dpl$ chooses the corresponding pebble
respectively from $A'_s, B'_s$ 
or $C'_s$,  obtaining a configuration $(i_q, j_q)$ with 
$i_q \in X_s$ iff $j_q \in X'_s$ for $s < h+1$ and $X \in \{A,B,C\}$.
$seg(i_q)=seg(j_q)=s+1$. In fact, this is an identical configuration. 
\item[(iii)] whenever $\spl$ keeps his pebbles on $A'_{h+1}, B'_{h+1}$ and $C'_{h+1}$, $\dpl$ 
keeps his pebbles on $A_h, B_h$ and $C_h$ respectively. 
If $\spl$ picks one of these pebbles from $A_h, B_h$ or $C_h$ then $\dpl$ chooses the corresponding pebble
respectively from $A'_{h+1}, B'_{h+1}$ 
or $C'_{h+1}$, obtaining a configuration $(i_q, j_q)$ with 
$seg(i_q)=h+1, seg(j_q)=h+2$ with $i_q \in X_h$ and $j_q \in X'_{h+1}$
for $X \in \{A,B,C\}$.  There is a lag of one segment here.
\end{itemize}

Cases (i) and (ii) have been explored before. Let us now explore case (iii), 
which gives the configuration $(i_q, j_q)$ with $i_q \in X_h, j_q \in X'_{h+1}$, $X \in \{A,B,C\}$.

\begin{enumerate}
\item 
Let $i_q \in A_h, j_q \in A'_{h+1}$. Clearly, if $\spl$ moves to $A'_d, B'_d$ or $C'_d$, for $d \geq h+1$,
$\dpl$ moves respectively to $A_{d-1}, B_{d-1}$ or $C_{d-1}$.  Pebbling and picking a pebble 
here will give rise to a configuration as seen in (i), (ii) or (iii) above. 
The interesting case is when $\spl$ moves from $j_q$ to a point in 
some $D'_d$. Since there are no $D$ blocks in $\rho_1$, 
$\dpl$ moves to a point in $C_{d-1}$. 
Note that this is possible since 
$y_{d-1~ 3}-x_{h-1~ 1} \in I$ iff 
      $y'_{d 4}-x'_{h 1} \in I$. 
      Further, after pebbling, when $\spl$ picks a pebble, $\dpl$ can either ensure an identical 
      configuration, or a  configuration 
      as in 3(d).
\item Let $i_q \in B_h, j_q \in B'_{h+1}$. Clearly, if $\spl$ moves to $A'_d, B'_d$ or $C'_d$, 
$\dpl$ moves respectively to $A_{d-1}, B_{d-1}$ or $C_{d-1}$.  Pebbling and picking a pebble 
here will give rise to a configuration as seen in (i), (ii) or (iii) above. 
The interesting case is when $\spl$ moves from $j_q$ to a point in 
some $D'_d$. Since there are no $D$ blocks in $\rho_1$, 
$\dpl$ moves to a point in $A_d$. 
Note that this is possible since 
$y_{d 1}-x_{h-1 2} \in I$ iff 
      $y'_{d 4}-x'_{h 2} \in I$. 
      Further, after pebbling, when $\spl$ picks a pebble, $\dpl$ can either ensure an identical 
      configuration, or a  configuration 
      as in 3(d).
The case when $i_q \in C_h, j_q \in C'_{h+1}$
 is similar to the above : $\dpl$ can either preserve the lag or move to 
 $A_d$ whenever $\spl$ moves to $D'_d$. This is possible 
      since  $y_{d 1}-x_{h-1 3} \in I$ iff 
      $y'_{d 4}-x'_{h 3} \in I$. 
\end{enumerate}
\end{enumerate}
Thus, the possible configurations are (i) identical configurations $(i_p, j_p)$ with $i_p \in X_h$ 
iff $j_p \in X'_h$ with $X \in \{A, B, C\}$, and no segment lag, or (ii) configurations with no segment lag 
of the form $(i_p, j_p)$ with $i_p \in A_h, j_p \in D'_h$, or (iii) configurations 
with lag of one segment of the form $(i_p, j_p)$ with $i_p \in X_h, j_p \in X'_{h+1}$ with 
$X \in \{A, B, C\}$. 
If $\spl$ always chooses bounded intervals (of length $\leq l$), 
then $\dpl$ respects his segment lag of 1, and the maximum number of segments that can be explored in either word 
is atmost $nl < K$. In this case, there are atleast
$K-pl$ $\geq$ $nlk+nl-pl$ $\geq$ $(n-p)(l+1)$ segments to the right of $\rho_1$ and $K-pl+1$ segments to the right of $\rho_2$ after $p$ rounds.

If $\spl$ chooses an unbounded interval of length $>l$ in any round, then $\dpl$ 
moves ahead only by $l+1$ segments. The pebbles 
$\spl$ drops in his $D'$ blocks can be accommodated by $\dpl$ 
in the  $A$ blocks of these $l+1$ segments since, the 
number of points in the $A$ blocks are much more than $2nkl$, and atmost $k$ pebbles are placed in a round. 
Having done this, 
$\dpl$ can either enforce an identical configuration,   
or obtain the configuration $(i_p, j_p)$ 
with $i_p \in A_h$ and $j_p \in D'_h$. 
Since we have seen that $\dpl$ can always replicate 
$\spl$'s move from configurations of this kind $(i_p, j_p)$
$i_p \in A_h$ and $j_p \in D'_h$,  
for the remaining $n-p$ rounds 
either we have $seg(i_q)=seg(j_q), q >p$, or 
$|seg(i_q)-seg(j_q)|=1$ for all $q>p$. 
Thus, whenever an unbounded interval is used, the segments match,
and $\dpl$ ensures that the maximum segments covered after any $p$ rounds is 
$\leq p(l+1)$. This ensures that after $n$ rounds, we cover atmost $n(l+1)$ segments on either word. 
Thus, $\dpl$ can always replicate moves of the $\spl$ and 
there are 
$K-p(l+1)$ $\geq$ $nlk+nl-pl-p$ $\geq$ $n-p +(nl-pl)=(n-p)(l+1)$ segments 
 to the right of each word after $p$ rounds for all $p \leq n$. 
\end{proof}

\section{Oversampling : Relevant Lemmas}
\label{app:onf}

If $\psi$ is over $\Sigma \cup X$, then the relativization of $\psi$ with respect to $\Sigma$ is denoted 
$ONF_{\Sigma}(\psi)$ \cite{time14} and defined inductively as follows: If $\psi=a \in \Sigma$, 
then $ONF_{\Sigma}(\psi)=(a \wedge \bigvee \Sigma)$. Likewise, if $\psi=\varphi_1 \until \varphi_2$, then 
	$ONF_{\Sigma}(\psi)$=$[(act \rightarrow ONF_{\Sigma}(\varphi_1))$ $\until (act \wedge ONF_{\Sigma}(\varphi_2))]$ where 
	$act= \bigvee \Sigma$. It can then be seen \cite{arxiv-time14} that for 
	$\zeta_1=ONF_{\Sigma_1}(\psi_1)$ and $\zeta_2=ONF_{\Sigma_2}(\psi_2)$, 
	with $\Sigma_1=\Sigma \cup X_1$, $\Sigma_2=\Sigma \cup X_2$ and disjoint $X_1, X_2$,  
	if 
	$\varphi_1=\exists \downarrow X_1. \zeta_1$ and 
	$\varphi_2=\exists \downarrow X_2. \zeta_2$, then 
	$\varphi_1 \wedge \varphi_2=\exists \downarrow (X_1 \cup X_2). (\zeta_1 \wedge \zeta_2)$.
The following lemmas are from \cite{arxiv-time14}. 
\begin{lemma}
\label {lem:boolclosedequis}
Consider  formulae $\varphi_1, \varphi_2$ built from $\Sigma$.   
Let $\psi_1, \psi_2$ be formulae built from $\Sigma \cup X_1$ and $\Sigma \cup X_2$ respectively.
Let $X=X_1 \cup X_2$, $\Sigma_i=\Sigma \cup X_i$ for $i=1,2$, and  $X_1 \cap X_2=\emptyset$.
Let $\zeta_1= ONF_{\Sigma_1}(\psi_1)$ and $\zeta_2= ONF_{\Sigma_2}(\psi_2)$. Then, 
 $\varphi_1 = \exists \downarrow X_1. \zeta_1$ and  
  $\varphi_2 = \exists \downarrow X_2. \zeta_2$ implies
  $\varphi_1 \wedge \varphi_2 = \exists \downarrow X. 
  (\zeta_1 \wedge \zeta_2)$.
  \end{lemma}

 \begin{lemma}
\label{lem:flatgen}
Let $\varphi \in \mathsf{CMTL}$ be built from $\Sigma$, and $W$ be the set of witness variables obtained 
while flattening $\varphi$. 
Then $\varphi=\exists \downarrow W. ONF_{\Sigma}(\varphi_{flat})$.
\end{lemma}

\section{Proof of Lemma \ref{remove-cntinf}}
\label{app:inf}

We want to express that the number of times $b$ is true in the 
region $[l, \infty)$ of any timed word $\rho$ is $\geq n$. \\

Let $\Lbag \varphi_1, \varphi_2, \zeta, n-1\Rbag$ denote the formula 
 $\varphi_1\until [\varphi_2 \wedge (\varphi_1 \until [\varphi_2 \wedge \dots
 (\varphi_1 \until \zeta)\dots])]$, where the depth of the nested until is $n-1$.
The formula $\psi=
\fut_{[l, \infty)}[b\wedge \Lbag\neg b, b, b, n-1 \Rbag]$ in $\mtl$ captures  this requirement. 
Clearly,  $\psi$ evaluates to true on timed words where there is a point 
in $[l, \infty)$ where $b$ is true, and continues to be true for atleast $n$ times.

\section{Proof of Lemma \ref{cnt-sp}(2)}
\label{app:tmtlsp}
 We construct the 
$(\Sigma \cup W, X)$-simple extension $\rho'$ 	
	from $\rho$ exactly as we did in Lemma \ref{cnt-sp}(1), using the same set of new propositions $X=\{b_0, \dots, b_n\}$. We use these propositions as counters in the same way as in Lemma \ref{cnt-sp}. Note that if the number of occurrences of $b$ in some segment of the timed word $\rho'$ is less than $n$ then at least one of the counters $b_i$  will be missing. A point $h$ should be marked with witness $a$ iff there exist a point $j>h$ with $y \in \sigma_j$, $t_j-t_h \in I$, $x \in \sigma_i$ for $h<i<j$, and the number of times $b$ has occurred in $[t_h, t_j]$ is less than $n$.
	 Checking the number of occurrences of $b$ to be $<n$ amounts to checking that at least one of the propositions from $X$ is missing from $[t_h, t_j]$. The formula
		$\lambda = \wB[a \leftrightarrow [\bigvee \limits_{k = 1}^n  (x \wedge \neg b_k)] \until_I y]$ 
		captures all positions where this is true; all such psoitions are marked $a$.  
		Thus the formula $\zeta = \delta \wedge \lambda$ is the required formula in $\mathsf{MTL}$.

   \section{Timed Propositional  Temporal Logic ($\mathsf{TPTL}$)} 
 \label{app:tptl}
  \noindent \emph{Syntax of  $\mathsf{TPTL}$}: 
$\varphi::=a (\in \Sigma)~|true~|\varphi \wedge \varphi~|~\neg \varphi~|
~\varphi \until \varphi~|~ \nx \varphi ~|~y.\varphi~|~y\in I$
where $y$ is a clock variable.  There is a finite set $\mathcal{C}$ of
clock variables  progressing at the same rate, 
 and $I$ is an interval of the form $<a,b>$ $a,b\in \N$ with $<\in \{(,[\}$ and $>\in \{ ],) \}$.
 $\mathsf{TPTL}$ is interpreted over  words in $T\Sigma^*$. 
 The truth of a formula is interpreted at a position
 $i\in \mathbb{N}$ along the word. For a timed word $\rho=(\sigma_1,t_1)\dots(\sigma_n,t_n)$, we define the satisfiability relation, $\rho, i, \nu \models \phi$ saying that  the formula $\phi$ is true at position
 $i$ of the timed word $\rho$ with a valuation $\nu$ of all the clock
 variables at $i$. $\nu(x)$ is the valuation of clock $x$. 
 The notation $\nu[x \leftarrow t_i]$ represents 
 replacing the valuation of $x$ with $t_i$.\\ 
  	$\rho, i, \nu \models a$   $\leftrightarrow$  $a \in \sigma_{i}$ and 
 	$\rho,i,\nu  \models \neg \varphi$  $\leftrightarrow$   $\rho,i,\nu \nvDash  \varphi$\\
 	$\rho,i,\nu \models \varphi_{1} \wedge \varphi_{2}$   $\leftrightarrow$   $\rho,i,\nu \models \varphi_{1}$ 
 	and $\rho,i,\nu \models\ \varphi_{2}$\\
 	$\rho,i,\nu \models x.\varphi $   $\leftrightarrow$  $\rho,i,\nu[x \leftarrow t_i] \models \varphi$ and 
	$\rho,i,\nu \models x \in I $   $\leftrightarrow$ $t_i - \nu(x) \in I$\\
	$\rho,i,\nu \models \nx \varphi $   $\leftrightarrow$  $\rho,i+1,\nu \models \varphi$\\
	$\rho,i,\nu \models \varphi_{1} \until \varphi_{2}$  $\leftrightarrow$  $\exists j > i$, 
	$\rho,j,\nu  \models \varphi_{2}$,  
	and  $\rho,k,\nu  \models \varphi_{1}$ $\forall$ $i < k < j$

\noindent $\rho$ satisfies $\phi$ denoted $\rho \models \phi$ iff $\rho,1,\bar{0}\models \phi$. Here $\bar{0}$ 
is the valuation obtained by setting all clock variables to 0. $\mathsf{TPTL}^{\mathsf n}$ denotes the class of   $\mathsf{TPTL}$ formulae using  $\leq n$ clocks.   For example, 
 $x.[a \wedge \fut(b \wedge x \in (0,1) \wedge \fut(b \wedge x \in (0,1)))]$ is a formula in  
 $\mathsf{TPTL}^1$ which specifies that there are two $b$'s within distance (0,1) from $a$.

\subsection{$\mathsf{CTMTL} \subset \mathsf{TPTL}^1$}
\label{app:c0}

$\mathsf{CTMTL} \subseteq \mathsf{TPTL}^1$: We show that we can encode 
both the $\mathsf{C}$ modality as well as $\mathsf{UT}$ modality in  
$\mathsf{TPTL}^1$. 
 Consider the $\mathsf{C}$ modality 
 $\cnt_I^{\geq n}\varphi$. Recall that this formula holds good at a point $i$ in a timed word
 iff  $\varphi$ evaluates to true $\geq n$ times in the interval $t_i+I$. 
  We capture this in $\mathsf{TPTL}^1$ as follows:
\begin{itemize}
 \item  Let $\Lbag \varphi_1, \varphi_2, \zeta, n-1\Rbag$ denote the formula 
 $\varphi_1\until [\varphi_2 \wedge (\varphi_1 \until [\varphi_2 \wedge \dots
 (\varphi_1 \until \zeta)\dots])]$, where the depth of the nested until is $n-1$. 
 The $\mathsf{TPTL}$ formula 
 ${x.\fut(x \in I \wedge (\varphi \wedge 
\Lbag \neg \varphi, \varphi, \varphi \wedge x \in I, n-1 \Rbag))}$ evaluates to true 
at a position $i$ in any timed word $\rho$  iff 
there is a position $j > i$ such that $t_j \in t_i+I$, 
 $\rho, j, \nu \models \varphi$ with $\nu \in I$, 
and there exists a point $k > j$ such that 
$\rho, k, \nu \models \varphi$ with   
 $\nu \in I$, and  there exist $n-2$ points 
  $j < i_1 < i_2 < \dots < i_{n-2}< k$ where $\varphi$ evaluates to true, and $\neg \varphi$ is true 
 in $(j, i_1), (i_1, i_2), \dots, (i_{n-2}, k)$. It is clear that the clock valuation at all these $n-2$ inbetween points satisfy $ \nu \in I$ (since  $\nu \in I$ at both $j, k$). Thus, we have obtained $n$ points in $t_i+I$ where $\varphi$ is true. Clearly, 
 this captures the semantics of $\cnt_I^{\geq n}\varphi$.  
 
 \item The modality $\cnt^{< n}_I \varphi$ is obtained by negating 
 $\cnt^{\geq n}_I \varphi$, while 
 $\cnt^{=n}_I \varphi$ is written as a conjunction of 
 $\cnt^{\geq n}_I \varphi$ and $\cnt^{\leq n}_I \varphi$. 
\end{itemize}
 
 Now we embed  the $\mathsf{UT}$ modality $\varphi_1 \until_{I, \eta}\varphi_2$
     in $\mathsf{TPTL}^1$. Using Lemma \ref{final}, we just have to show 
     that the counting modality   $\varphi_1 \until_{I, \eta}\varphi_2$ where $\eta$ 
     is free of conjunctions 
     can be expressed in $\mathsf{TPTL}^1$. Let $\eta=\#_{\psi} \geq c$. Then 
     the formula $x.(\Lbag \varphi_1 \wedge \neg \psi, \varphi_1 \wedge \psi, \varphi _1 \wedge \psi \wedge \{\varphi_1 \until (\varphi_2 \wedge x \in I)\}, c\Rbag)$ 
     is the formula in  $\mathsf{TPTL}^1$ that captures $\varphi_1 \until_{I, \eta}\varphi_2$ : clearly, 
     this formula evaluates to true at a point $i$ iff there is a position $j >i$ 
     such that at $j$, $\varphi_2 \wedge x \in I$ evaluates to true (note that $x$ was reset at $i$), 
     and all the way between $i$ and $j$, $\varphi_1$ evaluates to true. Further, we have a nested until 
     of depth $c$, which witnesses $n$ points $i < i_1 < \dots < i_n < j$, 
     such that $\varphi_1 \wedge \psi$ evaluates to true at each $i_j$, and 
     $\neg \psi$ evaluates to true in $(i_{j-1}, i_j)$. This process can be repeated to 
     handle threshold formulae of counting depth $i >1$, by recursively 
     replacing the threshold formulae at each level of $\eta$ by an appropriate 
      $\mathsf{TPTL}^1$ formula. 
      Finally, the untimed threshold counting modality 
      introduced in Lemma \ref{final} can be 
      replaced in  $\mathsf{TPTL}^1$  by a technique similar to that in \cite{cltl}.\\
 \vspace{.1cm}      
 
 To show the strict containment of $\mathsf{CTMTL}$ in $\mathsf{TPTL}^1$, 
 we consider the formula $\varphi=x. \fut[a \wedge x \in (0,1) \wedge \Box[x \in (0,1) \rightarrow \neg b]] \in \mathsf{TPTL}^1$. 
 The $\mathsf{TPTL}^1$ formula says that there is an $a$ in $(0,1)$, and 
the last symbol in (0,1) is not a $b$. 
 This formula was shown to be not expressible in $\mathsf{MTL}$ \cite{bouyer}. 
 We show here that $\varphi$ cannot be expressed even with counting, that is in $\mathsf{CTMTL}$. 
 We show that for any choice of $n$ rounds and $k$ pebbles, 
we can find two words $\rho_1, \rho_2$ such that 
 $\rho_2 \models \varphi, \rho_1 \nvDash \varphi$, but 
 $\rho_1 \equiv_{n,k}^{\mathsf{CTMTL}} \rho_2$. 
    Let $p \in \mathbb{N}$ be such that $pnk >>k$ and 
  $0<\delta < < \frac{1}{2pnk}$.
  \begin{enumerate}
 \item  Consider the word $\rho_1=((ab)^{pnk}(ab)^{pnk},\tau)$
  where the time stamps are as follows: the first $2pnk$ symbols 
  lie in the interval (0,1), with the first time stamp $t_1= 0.9$, 
  $t_{2pnk}= 0.9 + (2pnk-1)\delta$, $t_{i+1}-t_i=\delta$ for all $0 < i < 2pnk$. 
  The remaining $2pnk$ time stamps are such that $t_{2pnk+1}=1.1, 
  t_{4pnk}=1.1+(2pnk-1)\delta$ and $t_{i+1}-t_i=\delta$ for all $2pnk+1 < i < 4pnk$.
  By the choice of $\delta$, we have $1.1+ (2pnk-1)\delta < 1.2$, and
   $0.9 + (2pnk-1)\delta <1$.  

\item The second word is $\rho_2=((ab)^{pnk-1}a(ba)^{pnk}b, \tau')$, 
with time stamps $t'_i=t_i+\delta$ for  $1 \leq i \leq 2pnk-1$ and 
$t'_{2pnk}=1.1-\delta>1$, $t'_i=t_i$ for $2pnk+1 \leq i \leq 4pnk$.

 \item   In the case of $\rho_1$, the last time stamp $<1$ is  
   $t_{2pnk}=0.9 + (2pnk-1)\delta$
     and the letter at that position is 
  $\sigma_{2pnk}=b$. 
  For $\rho_2$,  $t'_{2pnk-1}=t_{2pnk-1}+\delta=t_{2pnk}=0.9+(2pnk-1)\delta$ is the last time stamp 
   $<1$, and the letter at this position is $\sigma'_{2pnk-1}=a$. 
    Hence, 
 $\rho_1 \nvDash \varphi, \rho_2 \models \varphi$. 
 \item While the last $b$ in (0,1) 
 of $\rho_1$ is at position $2pnk$ with time stamp 
 $0.9 + (2pnk-1)\delta$, 
   the last $b$ in (0,1) 
 of $\rho_2$ is at position $2pnk-2$ with time stamp 
  $0.9 + (2pnk-2)\delta$.

\end{enumerate}

We now show that in a $n,k$-$\mathsf{CTMTL}$ game over $\rho_1, \rho_2$, 
$\dpl$ wins. The main intuition here is that apart from the fact that there is a lag of one symbol
across intervals (0,1) and (1,2), there is no difference between $\rho_1$ and $\rho_2$.

From the initial configuration $(i_0, j_0)$ with time stamps (0,0), $\spl$ initiates an $\until_I$ move 
or a $\mathsf{C}$ move or a $\mathsf{UT}$ move.
\begin{enumerate}
\item 
If $\spl$ initiates a $\until_{(1,2)}$ move on $\rho_1$ and comes to 
the position $2pnk+1$ with time stamp 1.1 as part of the $\fut_{(1,2)}$ move, then $\dpl$ 
will come on $\rho_2$ to the position $2pnk+1$ with the same time stamp   
  $\sigma_{2pnk+1}=a=\sigma'_{2pnk+1}$. 
  So we have $i'_0=2pnk+1, j'_0=2pnk+1$.
  The future is identical in both words from this point. 
 The interesting case is when $\spl$ chooses to do the full until move or a $\mathsf{C}$ move
or a $\mathsf{UT}$ move at $(i'_0, j'_0)$.

 Consider the until move first. In this case, $\spl$ chooses some  position $1 < h \leq 2pnk$ 
 in $\rho_2$. In this case, $\dpl$ will choose the same position in $\rho_1$. Even though 
 the time stamps differ, all the points to the right in both $\rho_1, \rho_2$ lie in (0,1)
 from $t_h$. Moreover, the number of points to the right are the same, with the same symbols.
  Hence, $\dpl$ wins.  Now let us consider a $\mathsf{C}$ move from 
   $(i'_0, j'_0)$. 
      The only relevant move is $\cnt_{(0,1)}$, since all points to the right of $(i'_0, j'_0)$ lie in interval (0,1). 
  The number of points to the right of $(i'_0, j'_0)$ in both words are much larger than $k$. 
      The number of points between $i_0$ and $i'_0$ as well as $j_0$ and $j'_0$ are both much $>>k$; infact the number of $a$'s as well as $b$'s are much more than $k$. $\dpl$ can place his pebbles at the same positions as $\spl$, and obtain an identical configuration. 
The argument is exactly same for a $\mathsf{UT}$ move from  
  $(i'_0, j'_0)$.  Again, the only relevant move is $\until_{(0,1), \eta}$.

 \item   Let us now look at the more interesting case when $\spl$   
   initiates a $\until_{(0,1)}$ move or a $\cnt_{(0,1)}$ on $\rho_1$ from $(i_0, j_0)$ 
 and chooses  the last symbol in (0,1), the $b$ 
    at position $i'_0=2pnk$ of $\rho_1$ with time stamp $0.9+(2pnk-1)\delta$. 
        In this case, $\dpl$ will choose the last $b$ in (0,1) 
 at position $j'_0=2pnk-2$  of $\rho_2$ with time stamp $0.9+(2pnk-2)\delta$.
    In the case of $\until_{(0,1)}$ move, 
                $\spl$ can decide to end this move, in which case, the configuration will be
      $(i_1, j_1)$ with time stamps   
    $(0.9+(2pnk-1)\delta, 0.9+(2pnk-2)\delta)$. If $\spl$ decides to go ahead with the $\until$ move,
    and chooses a position $1 < h < 2pnk-2$ in $\dpl$'s word, then $\dpl$ will 
    pick  the same position $1 < h < 2pnk$ in $\spl$'s word. 
     This gives the identical configuration $(i_1, j_1)=(h, h)$. 
      All the points to the right of $h$ in $\rho_1$, as well as all the points 
      to the right of $h$ in $\rho_2$ lie in the interval (0,1), since the time stamps 
      are $t_h=0.9+(h-1)\delta$ and $t'_h=0.9+h \delta$. Clearly, any move of $\spl$ can be mimicked by $\dpl$ 
      obtaining an identical configuration. 
In case of the $\cnt_{(0,1)}$ move between $i_0$ and $i'_0$ and $j_0$ and $j'_0$,
 it can be seen that since the number of $a$'s and $b$'s are $>>k$, 
 $\dpl$ can place his pebbles at the same positions as $\spl$ and obtain an identical configuration. 
      
\item Now consider the case of a $\mathsf{UT}$ move. 
              Assume $\spl$ 
  initiates a $\until_{I, \eta}$ move with $I=(0,1)$ from (0,0), and
  plays on  $\rho_1$.  As part of the $\fut_{(0,1)}$ move, 
    If $\spl$ comes on the last position in (0,1) which is the $b$, 
  $\dpl$ will come on to the last  $b$ in (0,1). If the $\spl$ continues with the counting move, 
  then $\spl$ keeps $k$ pebbles in the positions between 0 and $2pnk$, while $\dpl$ keeps 
  his $k$ pebbles between 0 and $2pnk-2$ at identical positions in his own word. 
     It can be seen as in the case of the $\mathsf{C}$ move that $\dpl$ can 
    ensure an identical configuration. 
    
\item The argument when $\spl$ plays on $\rho_2$ is exactly the same.

  \end{enumerate}

\end{document}